\documentclass[11pt,reqno]{amsart}
\usepackage{amsthm, amsfonts, amssymb, color}
 \usepackage{mathrsfs}
\usepackage{enumerate}
\usepackage{amsmath}
 \usepackage{amstext, amsxtra}
  \usepackage{txfonts}
 \usepackage[colorlinks, linkcolor=black, citecolor=blue, pagebackref, hypertexnames=false]{hyperref}
\usepackage{graphicx}
\usepackage[symbol]{footmisc}
\usepackage{cite}
 \allowdisplaybreaks
\newtheorem{theorem}{Theorem}[section]
\newtheorem{proposition}[theorem]{Proposition}
\newtheorem{corollary}[theorem]{Corollary}
\newtheorem{lemma}[theorem]{Lemma}
\newtheorem{definition}[theorem]{Definition}

\newtheorem{remark}[theorem]{Remark}

\numberwithin{equation}{section}
\theoremstyle{definition}

\title[]
 {The scattering of fractional Schr\"{o}dinger operators with short range potentials}

\author{Rui Zhang, Tianxiao Huang\textsuperscript{*} and Quan Zheng}

\address{Rui Zhang, School of Mathematics and Statistics, Huazhong University of Science and Technology, Wuhan, Hubei 430074, PR China }
\email{lrzhang@hust.edu.cn}

\address{Tianxiao Huang (corresponding author), School of Mathematics (Zhuhai), Sun Yat-sen University, Zhuhai, Guangdong 519082, China}
\email{htx5@mail.sysu.edu.cn}

\address{Quan Zheng, School of Mathematics and Statistics, Huazhong University of Science and Technology, Wuhan, Hubei 430074, PR China }
\email{qzheng@hust.edu.cn}
\subjclass[2010]{35J10, 58J50}
\keywords{fractional Schr\"{o}dinger operator, short range potential, asymptotic completeness, limiting absorption principle}

\begin{document}

\begin{abstract}
For any positive real number $s$, we study the scattering theory in a unified way for the fractional Schr\"{o}dinger operator $H=H_0+V$, where $H_0=(-\Delta)^\frac s2$ and the real-valued potential $V$ satisfies short range condition. We prove the existence and asymptotic completeness of the wave operators $W_\pm=\mathrm{s-}\lim_{t\rightarrow\pm\infty}e^{itH}e^{-itH_0}$, the discreteness and finite multiplicity of the non-zero pure point spectrum $\sigma_\mathrm{pp}\setminus\{0\}$ of $H$, and the finite decay property of eigenfunctions. The short range condition is sharp with respect to the allowed decay rate of $V$, and the decay threshold for the existence and non-existence of the wave operators is faster than $|x|^{-1}$ at the infinity in some sense. Our approach is inspired by the theory of limiting absorption principle for simply characteristic operators established by S. Agmon and L. H\"{o}rmander in the 1970s.
\end{abstract}

\maketitle


\section{Introduction}\label{sec1}

Let $s$ be any positive real number, and $H_0=(-\Delta)^\frac s2$ be defined through the Fourier symbol $|\xi|^s$ which is self-adjoint in $L^2(\mathbb{R}^n)$, with domain $H^s$ and absolutely continuous spectrum. We consider the fractional Schr\"{o}dinger operator $H=H_0+V$, where $V$ is a real-valued multiplication operator which will be called a \textit{short range potential}. To formulate our results, we first introduce the subspace $B$ of $L^2(\mathbb{R}^n)$ defined by 
\begin{equation*}
B=\left\{u\in L^2(\mathbb{R}^n);~\|u\|_B\triangleq\sum_{j>0}R_j^\frac12\|u\|_{L^2(X_j)}<\infty\right\},
\end{equation*}
where 
\begin{equation}\label{equ12}
R_0=0,~R_j=2^{j-1}~\text{when}~j>0,~\text{and}~X_j=\{x\in\mathbb{R}^n;~R_{j-1}<|x|<R_j\}.
\end{equation}
$C_c^\infty$ and $\mathscr{F}^{-1}C_c^\infty$ are both dense in the Banach space $B$ where $\mathscr{F}^{-1}$ denotes the inverse Fourier transform. The dual space of $B$ is defined by
\begin{equation*}
B^*=\left\{v\in L_\mathrm{loc}^2(\mathbb{R}^n);~\|v\|_{B^*}\triangleq\sup_{j>0}R_j^{-\frac12}\|v\|_{L^2(X_j)}<\infty\right\},
\end{equation*} 
where the duality is given by the pairing $\langle v,u\rangle=\int_{\mathbb{R}^n}v\bar{u}dx$. We note that the embedding $\mathscr{S}\hookrightarrow B$ holds where $\mathscr{S}$ is the space of Schwartz functions, thus $B^*\hookrightarrow\mathscr{S}'$ where $\mathscr{S}'$ denotes the temperate distributions. Also note that $B$ is not reflexive and $\mathscr{S}$ is not dense in $B^*$, see \cite[p. 228]{Hor2}. When $s>0$, we define $B_s^*$ to be the space of $f\in\mathscr{S}'$ such that $(I-\Delta)^\frac s2f\in L_\mathrm{loc}^2$ and
\begin{equation}\label{14}
\|f\|_{B_s^*}\triangleq\|(I-\Delta)^\frac s2f\|_{B^*}<\infty.
\end{equation}
We have $B_s^*\hookrightarrow B^*$ (see Corollary \ref{cor42}) and thus elements of $B_s^*$ are $L_\mathrm{loc}^2$ functions.
\begin{definition}\label{def1}
Suppose $s>0$ and $V\in L_\mathrm{loc}^2(\mathbb{R}^n)$ is real-valued. We call $V$ a short range potential, if the map $f\mapsto Vf$ is compact from $B_s^*$ to $B$. 
\end{definition}
Note that $B\hookrightarrow L^2\hookrightarrow B^*$, a short range potential $V$ also defines a compact multiplication operator from $H^s$ to $L^2$, making the operator sum $H=H_0+V$ self-adjoint with domain $H^s$. The current work concerns some classical aspects in scattering theory. Our first main result is the existence of wave operators.

\begin{theorem}\label{thm51}
	The wave operators $W_\pm$ defined by the strong $L^2$ limits
	\begin{equation}\label{eq1.3}
	W_\pm u=\lim_{t\rightarrow\pm\infty}e^{itH}e^{-itH_0}u,\quad u\in L^2,
	\end{equation}
	exist and are isometric. Therefore, the following intertwining property holds:
	\begin{equation}
	e^{itH}W_\pm=W_\pm e^{itH_0},\quad t\in\mathbb{R}.
	\end{equation}
\end{theorem}

Let $\mathscr{H}_\mathrm{ac}$ and $\mathscr{H}_\mathrm{sing}$  respectively denote the absolutely continuous and singular continuous subspaces of $L^2(\mathbb{R}^n)$ with respect to $H$. Our second main result is the following.

\begin{theorem}\label{thm14}
	$W_\pm$ are asymptotically complete, that is,
	\begin{equation}\label{equ18}
	\mathrm{Ran}(W_+)=\mathrm{Ran}(W_-)=\mathscr{H}_\mathrm{ac}\quad\text{and}\quad\mathscr{H}_\mathrm{sing}=\{0\}.
	\end{equation}
	Consequently, the scattering operator $S=W_+^*W_-$ is unitary.
\end{theorem}

Let $\sigma_\mathrm{pp}$ denote the pure point spectrum of $H$, and $J_\tau=(I-\Delta)^\frac \tau 2$ for $\tau\in\mathbb{R}$. Our third main result regards the eigen properties of $H$.

\begin{theorem}\label{thm13}
$\sigma_\mathrm{pp}\setminus\{0\}$ is discrete in $\mathbb{R}\setminus\{0\}$. For $\lambda\in\sigma_\mathrm{pp}\setminus\{0\}$, the associated eigenspace is finite dimensional, and each eigenfunction $u$ satisfies
\begin{equation}\label{eq16}
\|\langle\cdot\rangle^{s-\epsilon}J_{s'}u\|_{L^2}<\infty,
\end{equation}
for any $\epsilon>0$ and $s'\leq s$, where $\langle\cdot\rangle=(1+|\cdot|^2)^\frac12$.
\end{theorem}

In the well known Schr\"{o}dinger case $s=2$, i.e. $H_0=-\Delta$, the term \textit{short range} physically means that $V$ decays fast enough to ensure a quantum scattering system to have locality in a large scale, for example, in regard to Theorem \ref{thm14}, that $e^{itH}$ asymptotically behaves as the  superposition of bound states $\in\mathscr{H}_\mathrm{pp}$ and a scattering state $\in\mathscr{H}_{ac}$ is what experiments have indicated. Researchers have found over the years that for $V$ to be a short range potential, the threshold of its decay rate at the infinity is the Coulomb type decay $|x|^{-1}$. For examples, the classical work \cite{D} of Dollard's shows that the wave operators do not exist if $V(x)=|x|^{-1}$; the physical background of Schr\"{o}dinger operators indicates that the embedded eigenvalues of $H=-\Delta+V$ should not occur when $V$ has short range, but Neumann-Wigner \cite{N} shows that this may not be true if $V(x)=O(|x|^{-1})$; closely relatedly, Kiselev \cite{Kis} has shown in dimension $n=1$ for the half-line Schr\"{o}dinger operators with Dirichlet boundary condition that such decay rate is sharp for the absence of singular spectrum. On the other hand for short range conditions on $V$, besides Reed-Simon \cite{RS3,RS4} and the fruitful references therein for such developed topics, we only mention a few papers which are more relevant to our work. Among early efforts, the peak of studying short range condition has been widely accepted to be the work \cite{A} of S. Agmon's, where it was proved that the compactness of the map
\begin{equation}\label{eq17}
H^2\ni u\mapsto(1+|\cdot|)^{1+\epsilon}Vu\in L^2
\end{equation}
indicates the existence and asymptotic completeness of the wave operators, mainly by establishing the so-called \textit{limiting absorption principle}:
\begin{equation}\label{eq1.8}
\sup_{\lambda>\lambda_0}\sup_{\epsilon>0}\|(H-\lambda\pm i\epsilon)^{-1}\|_{L^{2,\sigma}\rightarrow L^{2,-\sigma}}\leq C(V,\lambda_0),
\end{equation}
where $\lambda_0>0$, $\sigma>\frac12$ and $L^{2,\sigma}=\{f\in L^2;~(1+|\cdot|)^\sigma f\in L^2\}$. In an abstract point of view, since the spectral measure $dE_\lambda$ of $H$ is crucial to understanding the dynamics of $e^{itH}$, the importance of studying the resolvent of $H$ near the real axis comes from the formal fact that
\begin{equation}\label{dE}
dE_\lambda=\lim_{\epsilon\downarrow0}(2\pi i)^{-1}\left((H-\lambda-i\epsilon)^{-1}-(H-\lambda+i\epsilon)^{-1}\right).
\end{equation}
In particular, $|V(x)|\leq C(1+|x|)^{-1-\epsilon}$ verifies Agmon's condition, and such assumption is closely related to condition $\sigma>\frac12$ in the limiting absorption principle \eqref{eq1.8}. Notice that a special case of \eqref{eq1.8} is $H=H_0$ where the singularity of \eqref{dE} lies in the sphere $\{|\xi|^2=\lambda\}$, in Agmon-H\"{o}rmander \cite{AH}, such idea was pushed to a more critical situation where $L^{2,\sigma}\rightarrow L^{2,-\sigma}$ ($\sigma>\frac12$) is replaced by $B\rightarrow B^*$ by observing the sharp $L^2$ type Fourier restriction property $\mathscr{F}:~B\rightarrow L^2(\mathbb{S}^{n-1})$, while the short range scattering theory was correspondingly established in the Chapter XIV of H\"{o}rmander \cite{Hor2}. More recently Ionescu-Schlag \cite{IS}, inspired by the sharp $L^p$ Fourier restriction property (i.e. Stein-Tomas restriction theorem) $\mathscr{F}:~L^\frac{2n+2}{n+3}(\mathbb{R}^n)\rightarrow L^2(\mathbb{S}^{n-1})$ ($n>1$), as well as by an earlier effort Goldberg-Schlag \cite{go} in three dimensions, extended such ideas and established the limiting absorption principle in more appropriate interpolation spaces, resulting in a remarkable enlargement of the perturbation class for $-\Delta$ including first order perturbations.

In fact, the technicalities in Agmon \cite{A} had already worked for principal type differential operators. Moreover, Agmon-H\"{o}rmander \cite{AH} and H\"{o}rmander \cite{Hor2} had actually developed the short range scattering theory for the so-called simply characteristic operators, that is $P(D)$ defined by a real polynomial $P$ satisfying
\begin{equation}\label{eq110}
(\sum_{\alpha}|P^{(\alpha)}(\xi)|^2)^\frac12\leq C\left(1+|P(\xi)|+|\nabla P(\xi)|\right),\quad\xi\in\mathbb{R}^n,
\end{equation} 
where $P^{(\alpha)}(\xi)=\partial^\alpha P(\xi)$. A rough understanding of such condition is that, if we localize the Fourier symbol $(P(\xi)-\lambda)^{-1}$ in a unit ball centered at any $\eta\in\mathbb{R}^n$, then \eqref{eq110} ensures that such localization has some uniform property moduli the \textit{strength} of the translated $P$, i.e. the left hand side of \eqref{eq110} with $P^{(\alpha)}(\xi)$ replaced by $P^{(\alpha)}(\eta)$, while the restriction property $\mathscr{F}:~B\rightarrow L^2(M_\lambda)$ still holds where $M_\lambda=\{P(\xi)=\lambda\}$. Notice that the translation of $P(\xi)-\lambda$ generates a set of moduli polynomials that is precompact under the topology defined by strength, the localized resolvent estimates for $P$ can thus be pieced up to gain the resolvent estimate. The short range condition for a differential operator $V(x,D)$ is defined by the compactness of the map
\begin{equation*}
B_P^*\ni f\mapsto V(x,D)f\in B,
\end{equation*}
where
\begin{equation}\label{eq115}
\|f\|_{B_P^*}=\sum_\alpha\|P^{(\alpha)}(D)f\|_{B^*}.
\end{equation}
It turns out that the allowed decay for each coefficient of $V(x,D)$ is also slightly faster than $|x|^{-1}$ at the infinity (see \cite[p. 246]{Hor2}), however, the decay rate is only indicated by some integrability, which is in some sense better than the form of \eqref{eq17} where an $\epsilon$-decay appears. In particular, such theory applies to $P(D)=(-\Delta)^m$ where $m$ is any positive integer, and we note that in this case the $B_P^*$ norm in \eqref{eq115} is equivalent to the $B_{2m}^*$ norm that we define in \eqref{14}. (see Remark \ref{rrk42}.)

It is then natural to ask for the short range scattering theory of $(-\Delta)^\frac s2$ when $s>0$ is not an even number, and such study has already drawn some authors' attention, while understanding the dynamics of fractional Schr\"{o}dinger equations is of physical importance, see Laskin \cite{Las1,Las2}. Just to mention a few with respect to the research interests of the current paper, Giere \cite{g} considered functions of $-\Delta$ and obtained asymptotic completeness based on semigroup difference method, which applies to $(-\Delta)^\frac s2$ in the case $0<s<2$ perturbed by potentials decaying at the rate $(1+|x|)^{-1-\epsilon}$ in dimension $n>2+2\epsilon-s$. Kitada \cite{K} considered $s\geq1$ and obtained through eigenfunction expansion method that the asymptotic completeness for short range potentials verifying $|V(x)|\leq C(1+|x|)^{-1-\epsilon}$. (\cite{K} also obtained result for long range potentials.) There are also works specifically for $\sqrt{-\Delta}$ and relevant pseudo differential operators, see e.g. \cite{Sch,Umeda,Wei}.  

However, at the level of results, the above works in the fractional case still contain an $\epsilon$-decay assumption of $V$, unlike Agmon-H\"{o}rmander \cite{AH}, H\"{o}rmander \cite{Hor2} and Ionescu-Schlag \cite{IS} which have avoided such extra decay by considering better function spaces. Also, the range of order $s$ is limited in these works. 

The purpose of our work is to carry Agmon and H\"{o}rmander's approach to all fractional case $s>0$ for $(-\Delta)^\frac s2$, and obtain a sharper short range condition that is comparable to the case of $(-\Delta)^m$. A typical short range potential class by our Definition \ref{def1} is the collection of real-valued $V\in L_\mathrm{loc}^p$ satisfying (see Proposition \ref{prop25})
\begin{equation}\label{eq13}
p\begin{cases}
=\frac ns,\quad&\text{if}~0<s<\frac n2,\\
>2,\quad&\text{if}~s=\frac n2,\\
=2,\quad&\text{if}~s>\frac n2,
\end{cases}\quad\text{and}\quad \sum_{j>0}R_j\sup_{y\in X_j}\|V(\cdot+y)\|_{L^p(B(1))}<\infty,
\end{equation}
where $B(1)$ is the unit ball with center $0$. In particular, a real-valued $V$ satisfying
\begin{equation}\label{eq112}
|V(x)|\leq C(1+|x|)^{-1-\epsilon_j},\quad x\in X_j,~\epsilon_j\in(0,1),~\sum_{j>0}R_j^{-\epsilon_j}<\infty,
\end{equation}
verifies condition \eqref{eq13}, for example $\epsilon_j=1/j^{1-\epsilon}$ for any $0<\epsilon<1$. Moreover, the decay assumption \eqref{eq112} is sharp for the existence of wave operators by the following proposition (see Section \ref{sec4}):

\begin{proposition}\label{prop15}
Suppose $J_0\in\mathbb{N}_+$ and there is a sequence $\{\epsilon_j\}_{j>J_0}$ with $\epsilon_j\in[-1,+\infty)$ such that
\begin{equation}\label{eq114}
\sum_{j>J_0}R_j^{-\epsilon_j}=\infty.
\end{equation}
If $V\in L^\infty$ is real-valued and satisfies
\begin{equation}
\kappa_1(1+|x|)^{-1-\epsilon_j}\leq V(x)\leq \kappa_2(1+|x|)^{-1-\epsilon_j},\quad x\in X_j,\,j>J_0,
\end{equation}
or satisfies
\begin{equation}
-\kappa_2(1+|x|)^{-1-\epsilon_j}\leq V(x)\leq-\kappa_1(1+|x|)^{-1-\epsilon_j},\quad x\in X_j,\,j>J_0,
\end{equation}
for some $\kappa_2\geq \kappa_1>0$, then the wave operators \eqref{eq1.3} do not exist for all $s>0$.
\end{proposition}

A typical example for Proposition \ref{prop15} is taking $\epsilon_j=1/j$ in a subsequence. We also mention that recently, Ishida-Wada \cite[Theorem 1.6]{IW} has shown that the wave operators do not exist if $V(x)=\kappa\langle x\rangle^{-\gamma}$ for all $\kappa\in\mathbb{R}$, $\gamma\in(0,1]$ and $s>0$. Now the decay assumption \eqref{eq112} and Proposition \ref{prop15} give a sharper decay threshold for $V$ with respect to the existence and non-existence of the wave operators, and somehow surprisingly, such threshold decays faster than $|x|^{-1}$ at the infinity.

The framework of the proofs in this paper mainly follows that of H\"{o}rmander \cite[Chapter XIV]{Hor2} established for simply characteristic operators. However, when $s$ is not an even integer, there are two main new difficulties. 

First, the symbol $|\xi|^s$ of $H_0=(-\Delta)^\frac s2$ is not a polynomial. Notice that Agmon-H\"{o}rmander's theory for polynomial symbols $P$ highly relies on the fact that $P(D)$ is local, and on the simple characteristic condition \eqref{eq110} as well as the polynomial strength topology mentioned above, we have to find essential substitutions for such properties or concepts when considering $H_0$, and there will be two main kinds of arguments relevant. In the analysis of detailed characterization for $V$ to be short range in Section \ref{sec3}, we are indebted to the well known kernel estimate of the Bessel potential operator $(I-\Delta)^{-\frac s2}$ to substitute the locality property by some exponential decay tail estimate (see \eqref{eq411}), which will serve as an acceptable error in the boundedness and compactness discussions. In establishing the boundedness for boundary values $(H_0-\lambda\pm i0)^{-1}$ of the free resolvent (i.e. limiting absorption principle in the free case) in Section \ref{sec5} and Section \ref{sec6}, in order to piece up localized resolvent estimates, we use "ellipticity" of $|\xi|^s$ as already indicated in \eqref{14}, rather than the strength in \eqref{eq110}, to investigate compactness in some continuous type function spaces instead of under the polynomial strength topology which does not make sense anymore. 

The second difficulty is the limited smoothness of the symbol $|\xi|^s$ at $0$. Such a fact will be much relevant in Section \ref{sec6} to the study of weighted version of the free resolvent estimates, which shall be used to prove the eigen properties of $H=H_0+V$. As an interesting result, the eigenfunctions of $H$ can only be shown to have finite decay as \eqref{eq16} indicates. On the other hand, however, for any simply characteristic operator $P(D)$, every eigenfunction $u$ of $P(D)+V(x,D)$ with short range perturbation satisfies the rapid decay (see \cite[Corollary 14.5.6]{Hor2})
\begin{equation}\label{equ113}
\|\langle\cdot\rangle^rP^{(\alpha)}(D)u\|_{L^2}<\infty,
\end{equation} 
for all $r>0$ and $\alpha\in\mathbb{N}_0^n$, which is way more stronger than \eqref{eq16}. The Littlewood-Paley theory for inhomogeneous Lipschitz functions shall be used to treat the finite smoothness issue, due to which the reason for \eqref{eq16} to hold is (see Section \ref{sec7})
\begin{equation}\label{eq1.18}
\|\langle\cdot\rangle^{s+\frac12}J_su\|_{B^*}<\infty,
\end{equation}
where $u$ in priori is only assumed to be an eigenfunction in $B_s^*$, and \eqref{eq1.18} really lies in a clutch for us to show the genuine $L^2$ eigen properties of $H$.

The rest of this paper is organized as follows. In Section \ref{sec3}, we give equivalent conditions for a multiplication operator $V$ to be bounded and compact from $B_s^*$ to $B$. In Section \ref{sec4}, we give the proof of Theorem \ref{thm51} where potential $V$ is not necessarily required to be of short range (see Remark \ref{rk31}), and we also prove Proposition \ref{prop15}. In Section \ref{sec5}, we study the behavior of the free resolvent $R_0(z)=(H_0-z)^{-1}$ near the real axis except $0$ from above and below, whose boundary values $R_0(\lambda\pm i0)$ shall be realized as bounded operators from $B$ to $B_s^*$. In Section \ref{sec6}, we study the weighted estimates for $R_0(\lambda\pm i0)$ acting on functions whose Fourier transform vanish on the sphere $\{|\xi|^s=\lambda\}$, which will be used to study the eigen properties of $H$, and the Littlewood-Paley theory is used to handle the non-smoothness of $|\xi|^s$. In Section \ref{sec7}, we study a set $\Lambda\subset\mathbb{R}\setminus\{0\}$, which will be shown discrete and a subset of the non-zero pure point spectrum $\sigma_\mathrm{pp}\setminus\{0\}$ of $H$, while the associated eigenfunctions in a special form are proved to have finite regularity and finite decay. In Section \ref{sec8}, we construct distorted Fourier transforms to prove Theorem \ref{thm14}, where the discreteness of $\Lambda$ proved in Section \ref{sec7} is crucial for the construction, and we will see after the proof of Theorem \ref{thm14} that actually $\Lambda=\sigma_\mathrm{pp}\setminus\{0\}$ holds. In Section \ref{sec9}, we prove Theorem \ref{thm13} by applying the Fredholm theory in an appropriate space to claim that all eigenfunctions have the special form considered in Section \ref{sec7}.

The notations that we shall frequently use are the following. Let $\mathscr{F}f(\xi)=\hat{f}(\xi)=\int_{\mathbb{R}^n}e^{-i\xi\cdot x}f(x)dx$ denote the Fourier transform of $f$. We use $\chi_j$ to denote the characteristic function of $X_j$ defined in \eqref{equ12}. Let
\begin{equation*}
\tilde{X}_j=\overline{\bigcup_{|k-j|\leq1}X_k},
\end{equation*}
we use $\tilde{\chi}_j$ to denote the characteristic function of the interior of $\tilde{X}_j$. For any $\tau\in\mathbb{R}$, let $\langle\cdot\rangle^\frac \tau 2=(1+|\cdot|^2)^\frac \tau 2$ and $J_\tau=(I-\Delta)^\frac \tau 2=\mathscr{F}^{-1}(\langle\cdot\rangle^\tau)\ast$ acting on $\mathscr{S}'$, thus $\|v\|_{B_s^*}=\|J_sv\|_{B^*}$ if $s>0$ as mentioned before. We use
\begin{equation}\label{e115}
\langle f,g \rangle=\int_{\mathbb{R}^n}f\bar{g}dx
\end{equation}
to denote the duality pairing whenever $f\in B^*$, $g\in B$ or $f,g\in L^2$. Let $\mathbb{R}_+,\mathbb{N}_+$ be the positive real numbers and positive integers, $\mathbb{C}^+=\{z\in\mathbb{C};~\mathrm{Im}\,z\geq0\}$ and similarly define $\mathbb{C}^-$. We also use $B(r)$ to denote an open ball with radius $r$ and center $0$.

\section{Multiplication Operators from $B_s^*$ to $B$}\label{sec3}

In this section, we explore the equivalent conditions for $V$ to define a bounded and a compact multiplication operator respectively from $B_s^*$ to $B$, by analyzing its local behavior on each annulus $X_j$, and the main results are Theorem \ref{thm43} and Theorem \ref{thm44}. First recall that when $s>0$, the kernel $G_s=\mathscr{F}^{-1}(\langle\cdot\rangle^{-s})$ of the Bessel potential operator $J_{-s}$ satisfies (see e.g. \cite{gra2}) for some $a, C>0$ that
\begin{equation}\label{eq41}
\begin{cases}
G_s(x)\geq C^{-1},\quad&|x|\leq1,\\
0\leq G_s(x)\leq Ce^{-a|x|},&|x|>1.
\end{cases}
\end{equation}

\begin{lemma}\label{lm41}
Suppose $s>0$, $V\in L_\mathrm{loc}^2$ defines a bounded multiplication operator from $B_s^*$ to $B$, and denoted by $T_{j,k}=\chi_jVJ_{-s}\chi_k$. Then there exist $C,c>0$ such that
\begin{equation}\label{eq42}
\|T_{j,k}\|_{L^1(X_k)\rightarrow L^2(X_j)}\leq Ce^{-cR_{\max\{j,k\}}},\quad|j-k|\geq2.
\end{equation}
\end{lemma}

\begin{proof}
If $k-j\geq2$, $x\in X_j$ and $y\in X_k$, then we have $x-y\in\tilde{X}_k$, thus by \eqref{eq41} that
\begin{equation*}
0\leq G_s(x-y)\leq Ce^{-aR_{k-2}}=Ce^{-\frac a4R_k}.
\end{equation*}
Therefore if $f\in L^1(X_k)$, then
\begin{equation}\label{eq44}
|T_{j,k}f(x)|=\left|V(x)\int_{X_k}G_s(x-y)f(y)dy\right|\leq Ce^{-\frac a4R_k}|V(x)|\|f\|_{L^1(X_k)},\quad x\in X_j.
\end{equation}
Take $\phi\in C_c^\infty$ with $\phi(x)\equiv1$ when $|x|<1$, and set $\phi_j(\cdot)=\phi(\frac{\cdot}{R_j})$, we have $\phi_j\equiv1$ on $X_j$, and by \eqref{eq44} that
\begin{equation}\label{eq45}
\begin{split}
\|T_{j,k}f\|_{L^2(X_j)}\leq&Ce^{-\frac a4R_k}\|V\phi_j\|_{L^2}\|f\|_{L^1(X_k)}\\
\leq&Ce^{-\frac a4R_k}\|V\phi_j\|_B\|f\|_{L^1(X_k)}\\
\leq&Ce^{-\frac a4R_k}\|V\|_{B_s^*\rightarrow B}\|J_s\phi_j\|_{B^*}\|f\|_{L^1(X_k)}.
\end{split}
\end{equation} 
Notice that
\begin{equation*}
\begin{split}
\|J_s\phi_j\|_{B^*}\leq\|J_s\phi_j\|_{L^2}\leq CR_j^\frac n2\leq CR_k^\frac n2,
\end{split}
\end{equation*}
\eqref{eq42} then holds when $k-j\geq2$, and the other case is parallel.
\end{proof}

\begin{corollary}\label{cor42}
If $s>0$, then $J_{-s}$ is bounded in $B^*$. Consequently, we have $B_{s'}^*\hookrightarrow B_s^*$ if $s'>s\geq0$.
\end{corollary}

\begin{proof}
The fact that that $J_{-s}$ maps $B^*$ to $L_\mathrm{loc}^2$ follows from the proof. First notice that \eqref{eq42} is still true if $V=\chi_j$, because $\|V\phi_j\|_{L^2}\leq CR_k^\frac n2$ holds in \eqref{eq45} obviously. Now it suffices to consider the proof of
\begin{equation}\label{eq47}
R_j^{-\frac12}\|J_{-s}u\|_{L^2(X_j)}\leq C\|u\|_{B^*},\quad j\geq1.
\end{equation}
By \eqref{eq42}, 
\begin{equation}
\begin{split}
\|J_{-s}u\|_{L^2(X_j)}\leq&\sum_{|k-j|\geq2}\|J_{-s}\chi_ku\|_{L^2(X_j)}+\sum_{|k-j|\leq1}\|J_{-s}\chi_ku\|_{L^2(X_j)}\\
\leq&C\sum_{|k-j|\geq2}e^{-cR_{\max\{j,k\}}}\|u\|_{L^1(X_k)}+\sum_{|k-j|\leq1}\|J_{-s}\chi_ku\|_{L^2(X_j)}.
\end{split}
\end{equation}
It is easy to deduce
\begin{equation*}
\sum_{|k-j|\geq2}e^{-cR_{\max\{j,k\}}}\|u\|_{L^1(X_k)}\leq\sum_{k\geq1}e^{-c'R_k}\|u\|_{B^*}\leq C\|u\|_{B^*}.
\end{equation*}
On the other hand,
\begin{equation*}
R_j^{-\frac12}\sum_{|k-j|\leq1}\|J_{-s}\chi_ku\|_{L^2(X_j)}\leq R_j^{-\frac12}\sum_{|k-j|\leq1}\|\chi_ku\|_{L^2(\mathbb{R}^n)}\leq C\|u\|_{B^*}.
\end{equation*}
Thus \eqref{eq47} holds.
\end{proof}

We now characterize the boundedness of multiplication operators from $B_s^*$ to $B$.

\begin{theorem}\label{thm43}
Let $s>0$, $V\in L_\mathrm{loc}^2$. Then $V$ defines a bounded multiplication operator from $B_s^*$ to $B$, if and only if there exist $c_1,c_2>0$ and $M_j\geq0$ such that for $j\geq1$ we have
\begin{equation}\label{eq411}
\|VJ_{-s}f\|_{L^2(X_j)}\leq M_j\left(\|f\|_{L^2(\tilde{X}_j)}+e^{-c_2R_j}\|e^{-c_1|\cdot|}f\|_{L^1(\tilde{X}_j^c)}\right),\quad f\in B^*,
\end{equation}
and
\begin{equation}\label{eq412}
\sum_{j>0}M_jR_j<\infty.
\end{equation}
\end{theorem}

\begin{remark}
Note that in \eqref{eq411}, $f\in B^*$ implies $J_{-s}f\in L_\mathrm{loc}^2$ by Corollary \ref{cor42}, and the right hand side is obviously finite.
\end{remark}

\begin{proof}[Proof of Theorem \ref{thm43}]
Let $T_j=\chi_jVJ_{-s}$ for notation convention.
	
For the necessity, we first show the existence of $c_1,c_2$ and $M_j$ for \eqref{eq411}. Write
\begin{equation*}
T_jf=T_j\tilde{\chi}_jf+\sum_{|k-j|\geq2}T_{j,k}\chi_kf.
\end{equation*}
We first have
\begin{equation*}
\|T_j\tilde{\chi}_jf\|_{L^2(X_j)}\leq R_j^{-\frac12}\|VJ_{-s}\tilde{\chi}_jf\|_B\leq R_j^{-\frac12}\|V\|_{B_s^*\rightarrow B}\|\tilde{\chi}_jf\|_{B^*}\leq CR_j^{-1}\|f\|_{L^2(\tilde{X}_j)}.
\end{equation*}
By Lemma \ref{lm41}, we deduce
\begin{equation*}
\begin{split}
\sum_{|k-j|\geq2}\|T_{j,k}\chi_kf\|_{L^2(X_j)}\leq&C\sum_{|k-j|\geq2}e^{-cR_{\max\{j,k\}}}\|f\|_{L^1(X_k)}\\
\leq&C\sum_{|k-j|\geq2}e^{-cR_{\max\{j,k\}}+\frac c2R_k}\|e^{-\frac c2|\cdot|}f\|_{L^1(X_k)}\\
\leq&Ce^{-\frac c4R_j}\|e^{-\frac c2|\cdot|}f\|_{L^1(\tilde{X}_j^c)}.
\end{split}
\end{equation*}
Thus $c_1,c_2,M_j$ exist. Now we fix $c_1=\frac c2$ and $c_2=\frac c4$ above, and let $M_j$ be the least possible number for \eqref{eq411} to hold. Take $f_j\in B^*$ such that
\begin{equation}\label{eq416}
\begin{cases}
\|f_j\|_{L^2(\tilde{X}_j)}+e^{-\frac c4R_j}\|e^{-\frac c2|\cdot|}f\|_{L^1(\tilde{X}_j^c)}=1,\\
\|T_jf_j\|_{L^2(X_j)}\geq\frac12M_j.
\end{cases}
\end{equation}
Take $F=\sum\limits_{j\geq1}R_{3j}^\frac12\tilde{\chi}_{3j}f_{3j}$, we have
\begin{equation*}
\begin{split}
\|F\|_{B^*}=&\sup_{k\geq1}R_k^{-\frac12}\left\|\sum_{j\geq1}R_{3j}^\frac12\tilde{\chi}_{3j}f_{3j}\right\|_{L^2(X_k)}\\
=&\sup_{k\geq1}\max\left\{R_{3k-1}^{-\frac12}R_{3k}^\frac12\|f_{3k}\|_{L^2(X_{3k-1})},R_{3k}^{-\frac12}R_{3k}^\frac12\|f_{3k}\|_{L^2(X_{3k})},R_{3k+1}^{-\frac12}R_{3k}^\frac12\|f_{3k}\|_{L^2(X_{3k+1})}\right\}\\
\leq&\sqrt{2},
\end{split}
\end{equation*}
and therefore
\begin{equation}\label{eq418}
\sum_{k\geq1}R_k^\frac12\|T_kF\|_{L^2(X_k)}=\|VJ_{-s}F\|_B\leq\sqrt{2}\|V\|_{B_s^*\rightarrow B}<\infty.
\end{equation}
Consider a part of the left hand side of \eqref{eq418},
\begin{equation}\label{eq419}
\begin{split}
\infty>&\sum_{k\geq1}R_{3k}^\frac12\|T_{3k}F\|_{L^2(X_{3k})}=\sum_{k\geq1}R_{3k}^\frac12\left\|T_{3k}\sum_{j\geq1}R_{3j}^\frac12\tilde{\chi}_{3j}f_{3j}\right\|_{L^2(X_{3k})}\\
\geq&\sum_{k\geq1}R_{3k}\|T_{3k}\tilde{\chi}_{3k}f_{3k}\|_{L^2(X_{3k})}-\sum_{k\geq1}R_{3k}^\frac12\sum_{j\neq k}R_{3j}^\frac12\|T_{3k}\tilde{\chi}_{3j}f_{3j}\|_{L^2(X_{3k})}\\
\geq&\sum_{k\geq1}R_{3k}\|T_{3k}f_{3k}\|_{L^2(X_{3k})}-\sum_{k\geq1}R_{3k}\|T_{3k}(1-\tilde{\chi}_{3k})f_{3k}\|_{L^2(X_{3k})}\\
&\quad-\sum_{k\geq1}R_{3k}^\frac12\sum_{j\neq k}R_{3j}^\frac12\|T_{3k}\tilde{\chi}_{3j}f_{3j}\|_{L^2(X_{3k})}\\
\geq&\frac12\sum_{k\geq1}M_{3k}R_{3k}-\sum_{k\geq1}R_{3k}\|T_{3k}(1-\tilde{\chi}_{3k})f_{3k}\|_{L^2(X_{3k})}-\sum_{k\geq1}R_{3k}^\frac12\sum_{j\neq k}R_{3j}^\frac12\|T_{3k}\tilde{\chi}_{3j}f_{3j}\|_{L^2(X_{3k})},
\end{split}
\end{equation}
where we have used \eqref{eq416} in the last inequality. Now by Lemma \ref{lm41} and \eqref{eq416}, we have
\begin{equation}\label{eq420}
\begin{split}
\sum_{k\geq1}R_{3k}\|T_{3k}(1-\tilde{\chi}_{3k})f_{3k}\|_{L^2(X_{3k})}\leq&\sum_{k\geq1}R_{3k}\sum_{|j-3k|\geq2}\|T_{3k,j}f_{3k}\|_{L^2(X_{3k})}\\
\leq&C\sum_{k\geq1}R_{3k}\sum_{|j-3k|\geq2}e^{-cR_{\max\{3k,j\}}}\|f_{3k}\|_{L^1(X_j)}\\
\leq&C\sum_{k\geq1}R_{3k}\sum_{|j-3k|\geq2}e^{-cR_{\max\{3k,j\}}+\frac c2R_j}\|e^{-\frac c2|\cdot|}f_{3k}\|_{L^1(\tilde{X}_{3k}^c)}\\
\leq&C\sum_{k\geq1}R_{3k}\sum_{|j-3k|\geq2}e^{-cR_{\max\{3k,j\}}+\frac{3c}{4}R_j}\\
<&\infty.
\end{split}
\end{equation}
We also have by \eqref{eq416} that
\begin{equation}\label{eq421}
\begin{split}
\sum_{k\geq1}R_{3k}^\frac12\sum_{j\neq k}R_{3j}^\frac12\|T_{3k}\tilde{\chi}_{3j}f_{3j}\|_{L^2(X_{3k})}\leq&\sum_{k\geq1}R_{3k}^\frac12\sum_{j\neq k}R_{3j}^\frac12\sum_{l=-1}^{1}\|T_{3k,3j+l}f_{3j}\|_{L^2(X_{3k})}\\
\leq&C\sum_{k\geq1}R_{3k}^\frac12\sum_{j\neq k}R_{3j}^\frac12\sum_{l=-1}^{1}e^{-cR_{\max\{3k,3j+l\}}}\|f_{3j}\|_{L^1(X_{3j+l})}\\
\leq&C\sum_{k\geq1}R_{3k}^\frac12\sum_{j\neq k}R_{3j}^\frac{1+n}{2}\sum_{l=-1}^{1}e^{-cR_{\max\{3k,3j+l\}}}\|f_{3j}\|_{L^2(\tilde{X}_{3j})}\\
<&\infty.
\end{split}
\end{equation}
Combining \eqref{eq419}, \eqref{eq420} and \eqref{eq421}, we know that
\begin{equation*}
\sum_{k\geq1}M_{3k}R_{3k}<\infty.
\end{equation*}
Similarly consider $F=\sum_{j\geq0}R_{3j+l}^\frac12\tilde{\chi}_{3j+l}f_{3j+l}$, $l=1,2$, we obtain \eqref{eq412}.

For the sufficiency, if $u\in B_s^*$, we set $f=J_su\in B^*$, then by \eqref{eq411},
\begin{equation}\label{eq423}
\begin{split}
\|Vu\|_B=&\|VJ_{-s}f\|_B=\sum_{j\geq1}R_j^\frac12\|T_jf\|_{L^2(X_j)}\\
\leq&\sum_{j\geq1}M_jR_j^\frac12\left(\|f\|_{L^2(\tilde{X}_j)}+e^{-c_2R_j}\|e^{-c_1|\cdot|}f\|_{L^1(\tilde{X}_j^c)}\right)\\
\leq&3\sqrt{2}\sum_{j\geq1}M_jR_j\|f\|_{B^*}+\sum_{j\geq1}M_jR_j^\frac12e^{-c_2R_j}\|e^{-c_1|\cdot|}f\|_{L^1(\tilde{X}_j^c)}.
\end{split}
\end{equation}
Further, by Cauchy-Schwarz inequality it is easy to see that
\begin{equation}\label{eq424}
\begin{split}
\|e^{-c_1|\cdot|}f\|_{L^1(\tilde{X}_j^c)}\leq&\sqrt{2}\sum_{k\geq1}R_k^{-\frac12}\left\||\cdot|^\frac12e^{-c_1|\cdot|}f\right\|_{L^1(X_k)}\\
\leq&\sqrt{2}\sum_{k\geq1}\left\||\cdot|^\frac12 e^{-c_1|\cdot|}\right\|_{L^2(X_k)}\|f\|_{B^*}\\
\leq&C\|f\|_{B^*}.
\end{split}
\end{equation}
Then by \eqref{eq412}, we know from \eqref{eq423} and \eqref{eq424} that $V$ defines a bounded multiplication operator from $B_s^*$ to $B$.
\end{proof}

Next, we give a characterization of the compactness for multiplication operators from $B_s^*$ to $B$.

\begin{theorem}\label{thm44}
Suppose $s>0$, and $V$ defines a bounded multiplication operator from $B_s^*$ to $B$. Then the operator $V$ is compact, if and only if for any compact $K\subset\mathbb{R}^n$, the set
\begin{equation}\label{eq425}
\left\{Vu;~u\in L^2,~\mathrm{supp}\,u\subset K,~\|J_su\|_{L^2}<1\right\}
\end{equation}
is precompact in $L^2$.
\end{theorem}

\begin{proof}
The necessity is obvious. For the sufficiency, let $\{u_l\}_{l\geq1}\subset B_s^*$ and $\|u_l\|_{B_s^*}\leq1$ for all $l$. Take $\phi\in C_c^\infty(B(1))$ with $\phi\equiv1$ in $B(\frac23)$, and set $\phi_k(\cdot)=\phi(\frac{\cdot}{R_k})$, then we may write for any $k$ that
\begin{equation}\label{eq426}
Vu_l=(1-\phi_k)Vu_l+\phi_kVJ_{-s}(1-\phi_{k+2})J_su_l+V\phi_kJ_{-s}\phi_{k+2}J_su_l.
\end{equation}
First by \eqref{eq411}, \eqref{eq424} and that fact that $\mathrm{supp}\,(1-\phi_k)\subset\cup_{j\geq k}\overline{X}_j$,
\begin{equation}
\|(1-\phi_k)Vu_l\|_B\leq\sum_{j\geq k}CM_j(R_j+R_j^\frac12e^{-c_2R_j})\|u_l\|_{B_s^*}\leq C\sum_{j\geq k}M_jR_j,
\end{equation}
we then know by \eqref{eq412} that
\begin{equation}\label{eq428}
\lim_{k\rightarrow\infty}\sup_{l>0}\|(1-\phi_k)Vu_l\|_B=0.
\end{equation}
For the second term on the right hand side of \eqref{eq426}, notice that $j\leq k$ implies $\tilde{X}_j\cap\mathrm{supp}\,(1-\phi_{k+2})=\emptyset$, we then use \eqref{eq411} again to obtain
\begin{equation*}
\begin{split}
\|\phi_kVJ_{-s}(1-\phi_{k+2})J_su_l\|_B\leq&\sum_{j\leq k}M_jR_j^\frac12e^{-c_2R_j}\|e^{-c_2|\cdot|}(1-\phi_{k+2})J_su_l\|_{L^1(\tilde{X}_j^c)}\\
\leq&C\sum_{j\leq k}M_jR_j^\frac12e^{-c_2R_j}\sum_{\nu\geq k+2}e^{-\frac{c_2}{2}R_\nu}R_\nu^\frac{n+1}{2}R_\nu^{-\frac12}\|J_su_l\|_{L^2(X_\nu)}\\
\leq&C\sum_{j\leq k}M_jR_j^\frac12e^{-c_2R_j}e^{-c_2R_k}\sum_{\nu\geq1}e^{-\frac{c_2}{4}R_\nu}R_\nu^\frac{n+1}{2}\|u_l\|_{B_s^*}\\
\leq&Ce^{-c_2R_k},
\end{split}
\end{equation*}
and thus
\begin{equation}\label{eq430}
\lim_{k\rightarrow\infty}\sup_{l\geq1}\|\phi_kVJ_{-s}(1-\phi_{k+2})J_su_l\|_B=0.
\end{equation}
For the last term $VU_{k,l}$ on the right hand side of \eqref{eq426} where $U_{k,l}=\phi_kJ_{-s}\phi_{k+2}J_su_l$, we have $\mathrm{supp}\,U_{k,l}\subset B(R_k)$, and
\begin{equation*}
\|J_sU_{k,l}\|_{L^2}\leq C_k\|\phi_{k+2}J_su_l\|_{L^2}\leq C_k'\|u_l\|_{B_s^*}\leq C_k',\quad l\geq1,
\end{equation*}
where we have used the facts that $\phi_k\in C_c^\infty$ and $\|J_s\phi_kJ_{-s}\|_{L^2\rightarrow L^2}\leq C_k$. By the assumption on \eqref{eq425}, $\{VU_{k,l}\}_{l\geq1}$ has an $L^2$ convergent subsequence for any fixed $k$. We may just assume $\{VU_{k,l}\}_{l\geq1}$ is $L^2$ convergent for every $k$ by taking the diagonal, and since $VU_{k,l}$ is supported in $B(R_k)$ for all $l$, $\{VU_{k,l}\}_{l\geq1}$ is also $B$ convergent for every $k$. Now $\{Vu_l\}_{l\geq1}$ is a Cauchy sequence in $B$ by the uniform smallness \eqref{eq428} and \eqref{eq430}, which completes the proof.
\end{proof}

Now we can verify the sufficient condition \eqref{eq13} mentioned in the Introduction.

\begin{proposition}\label{prop25}
Suppose $s>0$ and $V\in L_\mathrm{loc}^p$ satisfies \eqref{eq13}. Then $V$ defines a compact multiplication operator from $B_s^*$ to $B$.
\end{proposition}

\begin{proof}
Let $M_j=\sup_{y\in X_j}\|V(\cdot+y)\|_{L^p(B(1))}$. We first prove the boundedness of $V$ in the view of Theorem \ref{thm43}, by showing the existence of $C,c_1,c_2>0$ such that
\begin{equation}\label{eq433}
\|VJ_{-s}f\|_{L^2(X_j)}\leq CM_j\left(\|f\|_{L^2(\tilde{X}_j)}+e^{-c_2R_j}\|e^{-c_1|\cdot|}f\|_{L^1(\tilde{X}_j^c)}\right),\quad f\in B^*.
\end{equation}
In the sequel, we always assume $k,k'\in\mathbb{Z}^n/\sqrt{n}$, and take $\varphi_k$ to be the characteristic function of $Q_k=\{x\in\mathbb{R}^n;~\sup_{1\leq i\leq n}|x_i-k_i|\leq\frac{1}{2\sqrt{n}}\}$. Write
\begin{equation}\label{eq434}
\|VJ_{-s}f\|_{L^2(X_j)}\leq\|VJ_{-s}\tilde{\chi}_jf\|_{L^2(X_j)}+\|VJ_{-s}(1-\tilde{\chi}_j)f\|_{L^2(X_j)}.
\end{equation}

For the first right hand side term of \eqref{eq434}, denoted by $q=\frac{2p}{p-2}$, ($q=\infty$ if $p=2$,) we have by H\"{o}lder's inequality that
\begin{equation*}
\begin{split}
\|VJ_{-s}\tilde{\chi}_jf\|_{L^2(X_j)}^2\leq&\sum_k\|V\|_{L^p(Q_k\cap X_j)}^2\|J_{-s}\tilde{\chi}_jf\|_{L^q(Q_k\cap X_j)}^2\\
\leq&M_j^2\sum_k\left(\sum_{k'}\|J_{-s}\varphi_{k'}\tilde{\chi}_jf\|_{L^q(Q_k\cap X_j)}\right)^2.
\end{split}
\end{equation*}
Notice that
\begin{equation*}
\|J_{-s}\varphi_{k'}\tilde{\chi}_jf\|_{L^q(Q_k\cap X_j)}\leq\left\|\int_{Q_{k'}\cap\tilde{X}_j} G_s(x-y)|\varphi_{k'}(y)\tilde{\chi}_j(y)f(y)|dy\right\|_{L_x^q(Q_k\cap X_j)}.
\end{equation*}
If $|k-k'|\geq2$, $x\in Q_k$ and $y\in Q_{k'}$, then $|x-y|\geq1$ and $G_s(x-y)\leq Ce^{-a|k-k'|}$ holds by \eqref{eq41}, thus
\begin{equation}\label{eq437}
\|J_{-s}\varphi_{k'}\tilde{\chi}_jf\|_{L^q(Q_k\cap X_j)}\leq Ce^{-a|k-k'|}\|f\|_{L^1(Q_{k'}\cap\tilde{X}_j)}\leq Ce^{-a|k-k'|}\|f\|_{L^2(Q_{k'}\cap\tilde{X}_j)}.
\end{equation}
If $|k-k'|<2$, notice that
\begin{equation*}
q\begin{cases}
=\frac{2n}{n-2s},\quad&\text{if}~0<s<\frac n2,\\
>2,\quad&\text{if}~s=\frac n2,\\
=\infty,\quad&\text{if}~s>\frac n2,
\end{cases}
\end{equation*}
the Sobolev embedding $J_{-s}\in\mathscr{B}(L^2,L^q)$ shows that the last inequality of \eqref{eq437} is still true. Therefore
\begin{equation}\label{eq439}
\begin{split}
\|VJ_{-s}\tilde{\chi}_jf\|_{L^2(X_j)}\leq&M_j\left(\sum_k\left(\sum_{k'}Ce^{-a|k-k'|}\|f\|_{L^2(Q_{k'}\cap\tilde{X}_j)}\right)^2\right)^\frac12\\
\leq&CM_j\|f\|_{L^2(\tilde{X}_j)}
\end{split}
\end{equation}
holds for the inclusion $l^1\ast l^2\subset l^2$ is bilinearly continuous.

For the second right hand side term of \eqref{eq434}, 
\begin{equation}\label{eq440}
\begin{split}
\|VJ_{-s}(1-\tilde{\chi}_j)f\|_{L^2(X_j)}\leq&\sum_{|l-j|\geq2}\|VJ_{-s}\chi_lf\|_{L^2(X_j)}\\
\leq&C\sum_{|l-j|\geq2}\|V\|_{L^2(X_j)}e^{-cR_{\max\{j,l\}}}\|f\|_{L^1(X_l)}\\
\leq&CM_jR_j^\frac n2e^{-\frac c2R_j}\|e^{-\frac c2|\cdot|}f\|_{L^1(\tilde{X}_j^c)}\leq CM_je^{-\frac c4R_j}\|e^{-\frac c2|\cdot|}f\|_{L^1(\tilde{X}_j^c)},
\end{split}
\end{equation}
where in the second line we have used \eqref{eq42} as explained in the proof of Corollary \ref{cor42}, and in the last line we have used the facts $p\geq2$ and that $X_j$ can be covered by $O(R_j^\frac n2)$ many unit balls.

Now \eqref{eq433} is checked by \eqref{eq439} and \eqref{eq440}, thus $V$ is bounded from $B_s^*$ to $B$. For the compactness, it suffices to check Theorem \ref{thm44}. For any compact $K\subset\mathbb{R}^n$, consider
\begin{equation*}
\{u_l\}_{l\geq1}\subset\left\{u\in L^2;~\mathrm{supp}\,u\subset K,~\|J_su\|_{L^2}<1\right\}.
\end{equation*}
Of course $\{u_l\}_{l\geq1}$ is bounded in $L^2$, thus by Kolmogorov-Riesz theorem (see e.g. \cite{P}), we may assume $\{u_l\}_{l\geq1}$ is $L^2$ convergent up to subsequence. For any $N>0$, let $V=V_N+V_N'$, where
\begin{equation*}
V_N=\begin{cases}
V,\quad|V|\leq N,\\
0,\quad|V|>N.
\end{cases}
\end{equation*}
$\{V_Nu_l\}_{l\geq1}$ is also $L^2$ convergent for any fixed $N$. We assume $K\subset\overline{X_1\cup\cdots\cup X_J}$, then an argument similar to the boundedness part above shows that
\begin{equation*}
\begin{split}
\|V_N'u_l\|_{L^2}\leq&\sum_{j=1}^{J}\|V_N'J_{-s}J_su_l\|_{L^2(X_j)}\\
\leq&C\sum_{j=1}^J\sup_{y\in X_j}\|V_N'(\cdot+y)\|_{L^p(B(1))}\|J_su_l\|_{L^2}\\
\leq&C_J\|V_N'\|_{L^p(X_1\cup\cdots\cup X_{J+1})},
\end{split}
\end{equation*}
and therefore
\begin{equation*}
\lim_{N\rightarrow\infty}\sup_{l\geq1}\|V_N'u_l\|_{L^2}=0,
\end{equation*}
for $V\in L_\mathrm{loc}^p$ is assumed. Now clearly $\{Vu_l\}_{l\geq1}$ is an $L^2$ Cauchy sequence, which completes the proof.
\end{proof}

We end this section by a compact mapping property which will be used in Section \ref{sec7}.

\begin{lemma}\label{lm76}
Let $V$ define a compact multiplication operator from $B_s^*$ to $B$, $\{u_j\}_{j\geq1}$ be a bounded sequence in $B_s^*$, and $u_j\rightarrow u$ in $\mathscr{S}'$. Then $u\in B_s^*$ and $Vu_j\rightarrow Vu$ in $B$.
\end{lemma}

\begin{proof}	We may just assume $\|J_su_j\|_{B^*}\leq1$ for all $j$. For any $k\geq1$ and $\phi\in C_c^\infty(X_k)$, we have $\langle J_su_j,\phi\rangle\rightarrow\langle J_su,\phi\rangle$ since $J_s\phi\in\mathscr{S}$, and consequently
\begin{equation}\label{eq714}
|\langle J_su,\phi\rangle|\leq R_k^\frac12\|\phi\|_{L^2(X_k)}.
\end{equation}
Thus $u\in B_s^*$ with $\|u\|_{B_s^*}\leq1$. For the second statement, we may assume $u=0$. By \eqref{eq411} and $\|J_su_j\|_{B^*}\leq1$, it is easy to see for any $K>0$ that
\begin{equation}\label{eq715}
\|Vu_j\|_B\leq\sum_{k\leq K}R_k^\frac12\|Vu_j\|_{L^2(X_k)}+C\sum_{k>K}M_kR_k.
\end{equation}
Also notice that $Vu_j\rightarrow0$ in $L^2(X_k)$ for any fixed $k$, because if it does not hold, we can take a subsequence such that $\|Vu_{j_l}\|_{L^2(X_k)}\geq\epsilon$ for some $\epsilon>0$, and that $Vu_{j_l}\rightarrow g$ in $B$ by the compactness of $V$. However for any $\phi\in C_c^\infty(X_k)$, we have $V\phi\in L^2(X_k)$ for $V\in L_\mathrm{loc}^2$; we also know that $u_j\rightarrow0$ weakly in $L^2(X_k)$ by the assumption. Thus 
\begin{equation*}
\langle g,\phi\rangle=\lim_{l\rightarrow\infty}\int Vu_{j_l}\bar{\phi}dx=0,
\end{equation*}
i.e. $g=0$ in $X_k$, which is a contradiction. Now in the view of \eqref{eq412}, we know from \eqref{eq715} that $\|Vu_j\|_B\rightarrow0$, and the proof is complete.
\end{proof}

\section{Existence and Non-Existence of Wave Operators}\label{sec4}

The following proof of Theorem \ref{thm51} inherits the approach of \cite[Theorem 14.4.6]{Hor2}. Also see \cite{Hor76}.

\begin{proof}[Proof of Theorem \ref{thm51}]
The existence, isometry and the intertwining property follow if we can show the existence in a dense subset of $L^2$. By Cook's method (see e.g. \cite[p. 20]{RS3}), notice that $\mathscr{F}^{-1}C_c^\infty(\mathbb{R}^n\setminus\{0\})$ is a dense subset of $H^s$, it suffices to show for any $u\in\mathscr{F}^{-1}C_c^\infty(\mathbb{R}^n\setminus\{0\})$ that
\begin{equation}\label{eq53}
\int_{|t|>1}\|Ve^{-itH_0}u\|_{L^2}dt<\infty.
\end{equation}
First notice that
\begin{equation*}
0<r<\left|\nabla|\xi|^s\right|<R,\quad \xi\in\mathrm{supp}\,\hat{u},
\end{equation*}
for some $r,R>0$, we have
\begin{equation}\label{eq55}
|\partial_x^\alpha e^{-itH_0}u(x)|\leq C_{N,\alpha}(|x|+|t|)^{-N},\quad N\in\mathbb{N}_+,\,\alpha\in\mathbb{N}_0^n;\,|x|<\frac{r|t|}{2}~\text{or}~|x|>2R|t|.
\end{equation}
This is because of
\begin{equation*}
\partial_x^\alpha e^{-itH_0}u(x)=(2\pi)^{-n}\int e^{i(x\cdot\xi-t|\xi|^s)}\hat{u}(\xi)(i\xi)^\alpha d\xi,
\end{equation*}
and integration by part using the facts that $\xi\mapsto\frac{x\cdot\xi-t|\xi|^s}{|x|+|t|}$ has uniformly bounded derivatives in $\mathrm{supp}\,\hat{u}$, and has gradient uniformly bounded from below there.

Now take $\varphi\in C_c^\infty$ such that $\varphi(x)=1$ when $\frac r2<|x|<2R$ and $\varphi(x)=0$ when $|x|<\frac r4$ or $|x|>4R$. Let
\begin{equation*}
u_t^1(x)=(1-\varphi(x/t))e^{-itH_0}u(x),\quad u_t^2(x)=\varphi(x/t)e^{-itH_0}u(x).
\end{equation*}
For $u_t^1$, we have
\begin{equation*}
\|Vu_t^1\|_{L^2}\leq\|Vu_t^1\|_B\leq\|V\|_{B_s^*\rightarrow B}\|J_su_t^1\|_{B^*}\leq\|V\|_{B_s^*\rightarrow B}\|J_su_t^1\|_{L^2}.
\end{equation*}
We may take integer $m\geq\frac s2$, and use \eqref{eq55} together with Leibniz' formula to deduce
\begin{equation*}
\|J_su_t^1\|_{L^2}\leq\|(I-\Delta)^m(1-\varphi(\cdot/t))e^{-itH_0}u\|_{L^2}\leq C|t|^{-2}.
\end{equation*}
Therefore
\begin{equation*}
\int_{|t|>1}\|Vu_t^1\|_{L^2}dt\leq C\int_{|t|>1}|t|^{-2}dt<\infty.
\end{equation*}
For $u_t^2$, we have $r|t|/4<|x|<4R|t|$ in $\mathrm{supp}\,u_t^2$, then in the view of \eqref{eq411}, we have
\begin{equation*}
\begin{split}
\|Vu_t^2\|_{L^2}\leq&\sum_{r|t|/4<R_j<8R|t|}\|VJ_{-s}J_su_t^2\|_{L^2(X_j)}\\
\leq&\sum_{r|t|/4<R_j<8R|t|}M_j\left(\|J_su_t^2\|_{L^2(\tilde{X}_j)}+e^{-c_2R_j}\|e^{-c_1|\cdot|}J_su_t^2\|_{L^1(\tilde{X}_j^c)}\right)\\
\leq&C\sum_{r|t|/4<R_j<8R|t|}M_j\|J_su_t^2\|_{L^2(\mathbb{R}^n)}.
\end{split}
\end{equation*}
By taking integer $m\geq\frac s2$ and using Leibniz' formula, we also know from $\hat{u}\in C_c^\infty$ that
\begin{equation*}
\|J_su_t^2\|_{L^2}\leq\|(I-\Delta)^m\varphi(\cdot/t)e^{-itH_0}u\|_{L^2}\leq C,
\end{equation*}
and therefore
\begin{equation*}
\|Vu_t^2\|_{L^2}\leq C\sum_{r|t|/4<R_j<8R|t|}M_j\leq \frac{C}{r|t|}\sum_{r|t|/4<R_j<8R|t|}M_jR_j.
\end{equation*}
Finally by \eqref{eq412},
\begin{equation*}
\begin{split}
\int_{|t|>1}\|Vu_t^2\|_{L^2}dt\leq&\frac Cr\int_{|t|>1}\sum_{r|t|/4<R_j<8R|t|}\frac{M_jR_j}{|t|}dt\\
\leq&\frac Cr\sum_{j\geq1}\int_{R_j/8R<|t|<4R_j/r}\frac{M_jR_j}{|t|}dt\\
\leq&\frac Cr\log\frac{8R}{r}\sum_{j\geq1}M_jR_j\\
<&\infty.
\end{split}
\end{equation*}
We have now proved \eqref{eq53}, and the proof is complete.
\end{proof}

\begin{remark}\label{rk31}
In the proof of Theorem \ref{thm51}, we never use the assumption that $V$ is compact. In fact, if a real-valued $V$ defines a bounded multiplication operator from $B_s^*$ to $B$, and $H_0+V$ with domain $H^s$ is essentially self-adjoint, then Theorem \ref{thm51} holds for the closure $H$ of $H_0+V$ by exactly the same proof. For example, notice that for any $\epsilon>0$, we have for some $C_\epsilon>0$ that
\begin{equation*}
\|Vu\|_{L^2}\leq\|V\|_{B_s^*\rightarrow B}\|J_su\|_{L^2}\leq\|V\|_{B_s^*\rightarrow B}\left((1+\epsilon)\|H_0u\|_{L^2}+C_\epsilon\|u\|_{L^2}\right),\quad u\in H^s.
\end{equation*}
If $\|V\|_{B_s^*\rightarrow B}<1$, it follows by the Kato-Rellich theorem (see Reed-Simon \cite{RS2}) that $H=H_0+V$ is self-adjoint in $L^2$ with domain $H^s$, and consequently the wave operators exist.
\end{remark}

We next prove Proposition \ref{prop15}.

\begin{proof}[Proof of Proposition \ref{prop15}]
Assume that $W_+=s-\lim_{t\rightarrow+\infty}e^{itH}e^{-itH_0}$ exists. We first take $u\in\mathscr{F}^{-1}C_c^\infty(\mathbb{R}^n\setminus\{0\})$ with $\|u\|_{L^2}=1$ such that
\begin{equation*}
\frac56<\left|\nabla|\xi|^s\right|<\frac67,\quad\xi\in\mathrm{supp}\,\hat{u}.
\end{equation*}
Similar to \eqref{eq55}, we have
\begin{equation}\label{eq3.4}
|e^{-itH_0}u(x)|\leq C_N(|x|+|t|)^{-N},\quad N\in\mathbb{N}_+;\,|x|<\frac45|t|~\text{or}~|x|>|t|.
\end{equation}

We first claim that there exist $C_1,\,C_2>0$ such that
\begin{equation}\label{eq35}
\left|\int_{\mathbb{R}^n}V|e^{-itH_0}u|^2dx\right|\geq C_1R_j^{-1-\epsilon_j}-C_2|t|^{-2},\quad\frac54R_{j-1}<|t|<\frac45R_j,\,j>J_0.
\end{equation}
Since $V$ does not change sign, we deduce for $j>J_0$ that
\begin{equation*}
\begin{split}
\left|\int_{\mathbb{R}^n}V|e^{-itH_0}u|^2dx\right|&\geq\int_{X_j}|Ve^{-itH_0}u|^2dx\\
&\geq C_1R_j^{-1-\epsilon_j}\int_{X_j}|e^{-itH_0}u|^2dx\\
&\geq C_1R_j^{-1-\epsilon_j}-CR_j^{-1-\epsilon_j}\int_{\mathbb{R}^n\setminus X_j}|e^{-itH_0}u|^2dx\\
&\geq C_1R_j^{-1-\epsilon_j}-CR_j^{-1-\epsilon_j}\int_{\{|x|<\frac45|t|\}\cup\{|x|>\frac54|t|\}}|e^{-itH_0}u|^2dx\\
&\geq C_1R_j^{-1-\epsilon_j}-C_2|t|^{-2},
\end{split}
\end{equation*}
where in the last two steps we have used $\frac54R_{j-1}<|t|<\frac45R_j$, \eqref{eq3.4} and the assumption $\epsilon_j\geq-1$.

We next claim that there exist $C_3,C_4>0$ such that
\begin{equation}\label{eq37}
\left\|Ve^{-itH_0}u\right\|_{L^2}\leq C_3(R_{j-1}^{-1-\epsilon_{j-1}}+R_j^{-1-\epsilon_j})+C_4|t|^{-2},\quad R_{j-1}<|t|<R_j,\,j>J_0.
\end{equation}
To show this, we split the norm into
\begin{equation}\label{eq38}
\left\|Ve^{-itH_0}u\right\|_{L^2}\leq\left\|Ve^{-itH_0}u\right\|_{L^2(\{\frac12|t|<|x|<|t|\})}+\left\|Ve^{-itH_0}u\right\|_{L^2(\{|x|<\frac12|t|\}\cup\{|x|>|t|\})}.
\end{equation}
For the first right hand side term, when $R_{j-1}<|t|<R_j$ and $j>J_0$, we have
\begin{equation*}
\left\|Ve^{-itH_0}u\right\|_{L^2(\{\frac12|t|<|x|<|t|\})}\leq\left\|Ve^{-itH_0}u\right\|_{L^2(X_{j-1}\cup X_j)}\leq C_3(R_{j-1}^{-1-\epsilon_{j-1}}+R_j^{-1-\epsilon_j}).
\end{equation*}
For the second right hand side term of \eqref{eq38}, we again use \eqref{eq3.4} to get
\begin{equation*}
\left\|Ve^{-itH_0}u\right\|_{L^2(\{|x|<\frac12|t|\}\cup\{|x|>|t|\})}\leq C_4|t|^{-2},
\end{equation*}
and then \eqref{eq37} is shown.

Now we take $J_2>J_1>J_0$ such that
\begin{equation}\label{eq311}
\|e^{-itH_0}u-e^{-itH}W_+u\|_{L^2}=\|e^{itH}e^{-itH_0}u-W_+u\|_{L^2}\leq\frac{7C_1}{120C_3},\quad t\geq R_{J_1},
\end{equation}
and then
\begin{equation}\label{eq312}
\begin{split}
2&\geq\left|\langle e^{iR_{J_2}H}e^{-iR_{J_2}H_0}u-e^{iR_{J_1}H}e^{-iR_{J_1}H_0}u,W_+u\rangle\right|\\
&=\left|\int_{R_{J_1}}^{R_{J_2}}\frac{d}{dt}\langle e^{itH}e^{-itH_0}u,W_+u\rangle dt\right|\\
&=\left|\int_{R_{J_1}}^{R_{J_2}}\langle Ve^{-itH_0}u,e^{-itH}W_+u\rangle dt\right|\\
&\geq\left|\int_{R_{J_1}}^{R_{J_2}}\langle Ve^{-itH_0}u,e^{-itH_0}u\rangle dt\right|-\left|\int_{R_{J_1}}^{R_{J_2}}\langle Ve^{-itH_0}u,e^{-itH_0}u-e^{-itH}W_+u\rangle dt\right|\\
&\geq\int_{R_{J_1}}^{R_{J_2}}\left|\langle Ve^{-itH_0}u,e^{-itH_0}u\rangle\right| dt-\frac{7C_1}{120C_3}\int_{R_{J_1}}^{R_{J_2}}\|Ve^{-itH_0}u\|_{L^2}dt,
\end{split}
\end{equation}
where we again use that fact that $V$ does not change sign. Using \eqref{eq35}, we have
\begin{equation}\label{eq313}
\begin{split}
\int_{R_{J_1}}^{R_{J_2}}\left|\langle Ve^{-itH_0}u,e^{-itH_0}u\rangle\right| dt&\geq\sum_{j=J_1+1}^{J_2}\int_{\frac54R_{j-1}}^{\frac45R_j}\left|\langle Ve^{-itH_0}u,e^{-itH_0}u\rangle\right|dt\\
&\geq\sum_{j=J_1+1}^{J_2}\int_{\frac54R_{j-1}}^{\frac45R_j}(C_1R_j^{-1-\epsilon_j}-C_2t^{-2})dt\\
&\geq\frac{7C_1}{40}\sum_{j=J_1+1}^{J_2}R_j^{-\epsilon_j}-C.
\end{split}
\end{equation}
Using \eqref{eq37} and \eqref{eq311}, we have
\begin{equation}\label{eq3.14}
\begin{split}
&\frac{7C_1}{120C_3}\int_{R_{J_1}}^{R_{J_2}}\|Ve^{-itH_0}u\|_{L^2}dt\\
\leq&\frac{7C_1}{120C_3}\sum_{j=J_1+1}^{J_2}\int_{R_{j-1}}^{R_j}\left(C_3(R_{j-1}^{-1-\epsilon_{j-1}}+R_j^{-1-\epsilon_j})+C_4t^{-2}\right)dt\\
\leq&\frac{7C_1}{120}\sum_{j=J_1+1}^{J_2}\left(R_{j-1}^{-\epsilon_{j-1}}+\frac12R_j^{-\epsilon_j}\right)+C'\\
\leq&\frac{7C_1}{80}\sum_{j=J_1+1}^{J_2}R_j^{-\epsilon_j}+\frac{7C_1}{120}R_{J_1}^{-\epsilon_{J_1}}+C'.
\end{split}
\end{equation}
Combining \eqref{eq312}, \eqref{eq313} and \eqref{eq3.14}, we get
\begin{equation}
2\geq\frac{7C_1}{80}\sum_{j=J_1+1}^{J_2}R_j^{-\epsilon_j}-\frac{7C_1}{120}R_{J_1}^{-\epsilon_{J_1}}-C'',
\end{equation}
and then sending $J_2$ to $\infty$ yields a contradiction to \eqref{eq114}. Thus $W_+$ does not exist, and neither does $W_-$ by the same proof.
\end{proof}

\section{Free Resolvent Estimates: I}\label{sec5}

For $z\in\mathbb{C}\setminus\mathbb{R}$, let $R_0(z)=(H_0-z)^{-1}$ denote the usual $L^2$ resolvent of $H_0$, which is also defined via the Fourier multiplier $(|\xi|^s-z)^{-1}$. This section studies the behavior of $R_0(z)$ when $\mathrm{Im}\,z\rightarrow\pm 0$, and the main result is Theorem \ref{thm34}. Before getting start, we recall two gadgets that will be used in this and the next section. The first one is a Fourier multiplier theorem in $B$.

\begin{lemma}[{\cite[Corollary 14.1.5]{Hor2}}]\label{lm21}
	Let $r\in C^1(\mathbb{R}^n)$ with $|r|$ and $|\nabla r|$ bounded. Then
	\begin{equation}
	\|r(D)u\|_B\leq C(\sup_{\mathbb{R}^n}|r|+\sup_{\mathbb{R}^n}|\nabla r|)\|u\|_B,\quad u\in B,
	\end{equation}
	where $r(D)u=\mathscr{F}^{-1}(r\hat{u})$.
\end{lemma}
\begin{remark}\label{rrk42}
As mentioned in the Introduction, now an easy duality discussion shows that, if $m$ is a positive integer and $P(D)=(-\Delta)^m$, then $\|f\|_{B_{2m}^*}\approx\|(I+(-\Delta)^m)f\|_{B^*}\approx\|f\|_{B_P^*}$ where $B_P^*$ is defined by \eqref{eq115}.
\end{remark}

The second one shall be used to piece up localized estimates of the free resolvent.

\begin{lemma}[{\cite[Theorem 14.1.7]{Hor2}}]\label{lm22}
	Given $\phi\in C_c^\infty$, and set $\phi(D-\eta)u=\mathscr{F}^{-1}(\phi(\cdot-\eta)\hat{u})$ for $\eta\in\mathbb{R}^n$ and $u\in\mathscr{S}'$. Then
	\begin{equation}\label{eq22}
	\int\|\phi(D-\eta)u\|_B^2d\eta\leq C\|u\|_B^2,\quad u\in B.
	\end{equation}
	Moreover, if $\|\phi\|_{L^2}>0$, we have
	\begin{equation}\label{eq23}
	\|u\|_{B^*}^2\leq C\int\|\phi(D-\eta)u\|_{B^*}^2d\eta,\quad u\in L_{\mathrm{loc}}^2\cap\mathscr{S}'.
	\end{equation}
\end{lemma}

Now we start by a localized resolvent estimate for real valued functions. 

\begin{lemma}[{\cite[Theorem 14.2.2]{Hor2}}]\label{lm31}
	Let $X\subset\mathbb{R}^n$ be bounded and open, $\phi\in C_c^\infty(X)$, $C_1,C_2>0$, and given function $j:X\rightarrow\{1,\cdots,n\}$. Consider all $p\in C^2(X;\mathbb{R})$ such that
	\begin{equation}\label{eq31}
	|\partial_{j(\xi)}p(\xi)|\geq C_1\,\,\text{and}\,\,|\partial^\alpha p(\xi)|\leq C_2,\quad\xi\in\mathrm{supp}\,\phi,\,\,|\alpha|\leq2.
	\end{equation}
	Then for $f\in B$, the maps
	\begin{equation}\label{eq32}
	\mathbb{C}^\pm\ni z\mapsto u_z=\mathscr{F}^{-1}((p-z)^{-1}\phi\hat{f})\in B^*
	\end{equation}
	can be uniquely defined as a weak* continuous function, (i.e. $\mathbb{C}^\pm\ni z\mapsto\langle u_z,v\rangle\in\mathbb{C}$ are continuous for every $v\in B$,) satisfying
	\begin{equation}\label{eq33}
	\|u_z\|_{B^*}\leq C\|f\|_B,\quad f\in B,~z\in\mathbb{C}^\pm,
	\end{equation}
	where $C>0$ is uniform in the above $p$.
\end{lemma}

Before proceeding, it is necessary to further clarify the definition of $u_z$ in \eqref{eq32}. Recall the fact that if $p\in C^2$ is real-valued and $0$ is not a critical value of $p$, (i.e. $p=0\Rightarrow\nabla p\neq0$,) then
\begin{equation*}
(p\pm i0)^{-1}=\lim_{\epsilon\rightarrow+0}(p\pm i\epsilon)^{-1}
\end{equation*}
exist as first order distributions. Thus by assumption \eqref{eq31}, the maps \eqref{eq32} are well defined for $f\in\mathscr{S}$ when $z\in\mathbb{C}\setminus\mathbb{R}$, which in \cite{Hor2}, through an argument of partition of unity, were shown continuous as $\mathscr{S}'$-valued functions satisfying \eqref{eq33} and can be uniquely extended to $\mathbb{C}^\pm$. Since $\mathscr{S}$ is dense in $B$, such extensions are also well defined for $f\in B$ as $B^*$-valued weak* continuous functions satisfying \eqref{eq33}. 

\begin{lemma}\label{lm32}
	Let $X\subset\mathbb{R}^n$ be bounded and open, $\phi\in C_c^\infty(X)$, and $Q_0\in C(X;\mathbb{R})$. For some $\epsilon>0$, denoted by $K_1=\overline{\{\xi\in X;~|Q_0(\xi)|<4\epsilon\}}\cap\mathrm{supp}\,\phi$, we assume that $Q_0\in C^2(K_1)$ and $\nabla Q_0\neq0$ holds in $K_1$. Then there exist $\delta_0,C>0$, such that when $Q\in C(X;\mathbb{R})\cap C^2(K_1)$ satisfies
	\begin{equation}\label{eq34}
	\|Q-Q_0\|_{C^2(K_1)}+\|Q-Q_0\|_{C(K_2)}<\delta_0
	\end{equation}
	where $K_2=\overline{\{\xi\in X;~|Q_0(\xi)|>3\epsilon\}}\cap\mathrm{supp}\,\phi$, we have
	\begin{equation}
	\|\mathscr{F}^{-1}((Q-\zeta)^{-1}\phi\hat{f})\|_{B^*}\leq C\|f\|_B,\quad f\in B,~\zeta\in\mathbb{C}^\pm,~|\zeta|<\delta_0.
	\end{equation}
\end{lemma}

\begin{proof}
	By the assumption, we first note that $0$ is not a critical value of $Q_0$ in $\mathrm{supp}\,\phi$. This will lead to the fact that $0$ is not a critical value of $Q-\zeta$ in $\mathrm{supp}\,\phi$ as long as $|\zeta|$ is small and \eqref{eq34} holds, as we will show in the following.
	
	Take $h\in C_c^\infty(\mathbb{R})$ such that $h(s)=1$ when $|s|<3\epsilon$ and $h(s)=0$ when $s>4\epsilon$, we may write $\phi=\phi_1+\phi_2$ where
	\begin{equation}\label{eq36}
	\phi_1(\xi)=h(Q_0(\xi))\phi(\xi),\quad\phi_2(\xi)=(1-h(Q_0(\xi)))\phi(\xi).
	\end{equation}
	Obviously $\mathrm{supp}\,\phi_1\subset K_1$, and we may assume
	\begin{equation*}
	|\nabla Q_0(\xi)|\geq C_1\,\,\text{and}\,\,|\partial^\alpha Q_0(\xi)|\leq C_2,\quad\xi\in K_1,~|\alpha|\leq2.
	\end{equation*}
	Then for $\xi\in K_1$, there exists $j(\xi)\in\{1,\cdots,n\}$ such that $|\partial_{j(\xi)}Q_0(\xi)|\geq\frac{C_1}{\sqrt{n}}$, thus when
	\begin{equation*}
	\|Q-Q_0\|_{C^2(K_1)}<\delta_0\leq\mbox{$\frac{C_1}{2\sqrt{n}}$},
	\end{equation*}
	we have
	\begin{equation*}
	|\partial_{j(\xi)}Q(\xi)|\geq\mbox{$\frac{C_1}{2\sqrt{n}}$}\,\,\text{and}\,\,|\partial^\alpha Q(\xi)|\leq C_2+\mbox{$\frac{C_1}{2\sqrt{n}}$},\quad\xi\in K_1,~|\alpha|\leq2.
	\end{equation*}
	By Lemma \ref{lm31}, we obtain
	\begin{equation*}
	\|\mathscr{F}^{-1}((Q-\zeta)^{-1}\phi_1\hat{f})\|_{B^*}\leq C\|f\|_B,\quad f\in B,~\zeta\in\mathbb{C}^\pm.
	\end{equation*}
	
	On the other hand, when $\xi\in\mathrm{supp}\,\phi_2\subset K_2$, we have $|Q_0(\xi)|\geq3\epsilon$. Thus when $|\zeta|<\delta_0\leq\epsilon$ and
	\begin{equation*}
	\|Q-Q_0\|_{C(K_2)}<\delta_0,
	\end{equation*}
	we have $|Q-\zeta|\geq\epsilon$ in $\mathrm{supp}\,\phi_2$, and consequently
	\begin{equation*}
	\|\mathscr{F}^{-1}((Q-\zeta)^{-1}\phi_2\hat{f})\|_{B^*}\leq\|\mathscr{F}^{-1}((Q-\zeta)^{-1}\phi_2\hat{f})\|_{L^2}\leq\epsilon^{-1}\|f\|_{L^2}\leq\epsilon^{-1}\|f\|_B,
	\end{equation*}
	which completes the proof.
\end{proof}

\begin{lemma}\label{lm33}
	Given $z_0\in\mathbb{R}\setminus\{0\}$ and $\phi\in C_c^\infty(B(1))$. For $\eta\in\mathbb{R}^n$, we define
	\begin{equation}
	Q_\eta(\xi)=\frac{|\xi+\eta|^s-z_0}{\langle\eta\rangle^s},\quad\xi\in\mathbb{R}^n.
	\end{equation}
	Then there exist $\delta_0,C>0$ such that
	\begin{equation}\label{eq314}
	\|\mathscr{F}^{-1}((Q_\eta-\zeta)^{-1}\phi\hat{f})\|_{B^*}\leq C\|f\|_B,\quad f\in B,~\zeta\in\mathbb{C}^\pm,~|\zeta|<\delta_0,~\eta\in\mathbb{R}^n.
	\end{equation}
\end{lemma}

\begin{proof}
	As in the proof of Lemma \ref{lm32}, we will show through the following details that when $z_0\neq0$ and $|\zeta|$ is small, $0$ is not a critical value of $Q-\zeta$ in $\mathrm{supp}\,\phi$.
	
	First consider the set of functions $\mathcal{A}=\{Q_\eta;~|\eta|>2\}$. Obviously $\mathcal{A}\subset C^\infty(\overline{B(1)})$. Let $\overline{\mathcal{A}}$ be the closure of $\mathcal{A}$ in $C^2(\overline{B(1)})$ and suppose $Q\in\overline{\mathcal{A}}$. If $Q=Q_\eta$ for some $\eta$ with $|\eta|\geq2$, we have
	\begin{equation*}
	|\nabla Q(\xi)|=\frac{s|\xi+\eta|^{s-1}}{\langle\eta\rangle^s}\geq C_\eta>0,\quad\xi\in\overline{B(1)}.
	\end{equation*}
	If $Q=\lim_{j\rightarrow\infty}Q_{\eta_j}$ in $C^2(\overline{B(1)})$ for some $\eta_j\rightarrow\infty$, by the inequality
	\begin{equation*}
	\left||\xi|^s-z_0\right|+s|\xi|^{s-1}+1\geq C_{z_0,s}\mbox{$(\frac52+|\xi|^2)^\frac s2$},\quad\xi\in\mathbb{R}^n\setminus\{0\},
	\end{equation*}
	we know for all $j$ that
	\begin{equation*}
	\begin{split}
	\frac{\left||\xi+\eta_j|^s-z_0\right|}{\langle\eta_j\rangle^s}+\frac{s|\xi+\eta_j|^{s-1}}{\langle\eta_j\rangle^s}+\frac{1}{\langle\eta_j\rangle^s}\geq&C_{z_0,s}\left(\frac{\frac52+|\xi+\eta_j|^2}{1+|\eta_j|^2}\right)^\frac s2\\
	\geq&2^{-\frac s2}C_{z_0,s},\quad\xi\in\overline{B(1)},
	\end{split}
	\end{equation*}
	and obtain by sending $j$ to $\infty$ that
	\begin{equation*}
	|Q(\xi)|+|\nabla Q(\xi)|\geq C>0,\quad\xi\in\mathrm{supp}\,\phi,
	\end{equation*}
	which also holds in the previous case. Therefore $Q$ satisfies the condition of $Q_0$ in Lemma \ref{lm32}. On the other hand, one checks by the Arzel\`{a}-Ascoli theorem that $\overline{\mathcal{A}}$ is compact in $C^2(\overline{B(1)})$. Now we know \eqref{eq314} is true for $|\eta|>2$ if we employ a finite covering argument for $\overline{\mathcal{A}}$ with respect to the $C^2$ topology by using Lemma \ref{lm32}.
	
	Next consider $\mathcal{B}=\{Q_\eta;~|\eta|\leq2\}$. If $z_0>0$, we take $\epsilon=\frac{z_0}{8\langle\eta\rangle^s}$, then $|Q_\eta(\xi)|<4\epsilon$ implies $|\xi+\eta|^s>\frac{z_0}{2}$. Therefore $Q_\eta\in C^2(K_1)$ where $K_1$ is defined in Lemma \ref{lm32} with $Q_0$ replaced by $Q_\eta$, and
	\begin{equation*}
	|\nabla Q_\eta(\xi)|=\frac{s|\xi+\eta|^{s-1}}{\langle\eta\rangle^s}\geq C_\eta>0,\quad\xi\in\mathrm{supp}\,\phi.
	\end{equation*}
	This implies
	\begin{equation*}
	|Q_\eta(\xi)|+|\nabla Q_\eta(\xi)|\geq C_\eta>0,\quad\xi\in\mathrm{supp}\,\phi,
	\end{equation*}
	which is also true if $z_0>0$ for $Q_\eta$ never vanishes. Now we can apply Lemma \ref{lm32} with $Q_0$ replaced by $Q_\eta$. For $\delta_0$ obtained Lemma \ref{lm32}, by the boundedness of $\xi$ and $\eta$, it is easy to see that there exists $\delta>0$ such that $|\eta'-\eta|<\delta$ implies 
	\begin{equation*}
	\|Q_{\eta'}-Q_\eta\|_{C^2(K_1)}+\|Q_{\eta'}-Q_\eta\|_{C(K_2)}<\delta_0.
	\end{equation*}
	Thus a finite covering argument for $\{\eta\in\mathbb{R}^n;~|\eta|\leq2\}$ shows that \eqref{eq314} is also true for $\mathcal{B}$. Now the proof is complete.
\end{proof}

It is clear that if $\lambda\in\mathbb{R}\setminus\{0\}$, $\lambda$ is not a critical value of $|\xi|^s$ which is $C^\infty$ near $\{|\xi|^s=\lambda\}$, thus by Lemma \ref{lm31}, $R_0(\lambda\pm i0)f$ is well defined at least when $\hat{f}\in C_c^\infty$, and now we are about to prove the main result of this section. 

\begin{theorem}\label{thm34}
	Let $K$ be a closed subset of $\mathbb{C}^+$ or $\mathbb{C}^-$ such that $0\notin K$ and $\mathrm{Re}\,K$ is bounded. Suppose $r\in C^1(\mathbb{R}^n)$ with
	\begin{equation}\label{eq323}
	\sup_{\xi\in\mathbb{R}^n}\left|\frac{r(\xi)}{\langle\xi\rangle^s}\right|+\sup_{\xi\in\mathbb{R}^n}\left|\nabla\frac{r(\xi)}{\langle\xi\rangle^s}\right|<\infty.
	\end{equation}
	Then
	\begin{equation}\label{eq324}
	\|r(D)R_0(z)f\|_{B^*}\leq C_K\|f\|_B,\quad f\in\mathscr{F}^{-1}C_c^\infty,~z\in K.
	\end{equation}
	Moreover, the maps
	\begin{equation}\label{eq325}
	\mathbb{C}^\pm\setminus\{0\}\ni z\mapsto r(D)R_0(z)f\in B^*
	\end{equation}
	can be uniquely extended for $f\in B$ as $B^*$-valued weak* continuous functions. In particular, $R_0(z)$ maps $B$ into $B_s^*$.
\end{theorem}

\begin{proof}
	Similar to the argument after Lemma \ref{lm31}, the maps in \eqref{eq325} are $\mathscr{S}'$-valued continuous functions if $\hat{f}\in C_c^\infty$, thus the second statement comes from \eqref{eq324} and the fact that $\mathscr{F}^{-1}C_c^\infty$ is dense in $B$.
	
	For any $z_0\in\mathbb{R}\setminus\{0\}$, we first prove that \eqref{eq324} is true for $z\in\mathbb{C}^\pm$ with $|z-z_0|$ sufficiently small. Notice that $|z-z_0|<\epsilon$ implies $\left|\frac{z-z_0}{\langle\eta\rangle^s}\right|<\epsilon$ for $\eta\in\mathbb{R}^n$, thus we take $\phi\in C_c^\infty(B(1))$ with $\phi\equiv1$ in $B(\frac12)$, and apply Lemma \ref{lm33} with $\zeta=\frac{z-z_0}{\langle\eta\rangle^s}$ to have
	\begin{equation*}
	\begin{split}
	\left\|\mathscr{F}^{-1}\left(\left(\mbox{$\frac{|\cdot+\eta|^s-z}{\langle\eta\rangle^s}$}\right)^{-1}\phi\hat{f}\right)\right\|_{B^*}=&\left\|\mathscr{F}^{-1}\left(\left(\mbox{$\frac{|\cdot+\eta|^s-z_0}{\langle\eta\rangle^s}-\frac{z-z_0}{\langle\eta\rangle^s}$}\right)^{-1}\phi\hat{f}\right)\right\|_{B^*}\\
	\leq&C\|f\|_B,\quad f\in B,~\eta\in\mathbb{R}^n,~z\in\mathbb{C}^\pm,~|z-z_0|<\epsilon,
	\end{split}
	\end{equation*}
	for some $\epsilon>0$. If we take $\phi_0\in C_c^\infty(B(\frac12))$ with $\|\phi_0\|_{L^2}>0$, and replace the above $\hat{f}$ with $\phi_0r(\cdot+\eta)\hat{f}(\cdot+\eta)$ where $f\in\mathscr{F}^{-1}C_c^\infty$, we have $\phi_0r(\cdot+\eta)\hat{f}(\cdot+\eta)\in \mathscr{F}B$ by Lemma \ref{lm21} with the assumption $r\in C^1$, and then for $\eta\in\mathbb{R}^n$ that
	\begin{equation*}
	\begin{split}
	\langle\eta\rangle^s\|\phi_0(D-\eta)r(D)R_0(z)f\|_{B^*}\leq&C\|r(D)\phi_0(D-\eta)f\|_B\\
	\leq&C\left\|(I-\Delta)^\frac s2\phi_0(D-\eta)f\right\|_B\\
	=&C\left\|(I-\Delta)^\frac s2\phi(D-\eta)\phi_0(D-\eta)f\right\|_B,
	\end{split}
	\end{equation*}
	where in the second line we have used Lemma \ref{lm21} with the assumption \eqref{eq323}. One checks that
	\begin{equation*}
	\left|\frac{\langle\xi\rangle^s\phi(\xi-\eta)}{\langle\eta\rangle^s}\right|+\left|\nabla_\xi\frac{\langle\xi\rangle^s\phi(\xi-\eta)}{\langle\eta\rangle^s}\right|\leq C,\quad\xi,\eta\in\mathbb{R}^n,
	\end{equation*}
	we thus use Lemma \ref{lm21} again to obtain
	\begin{equation}\label{eq329}
	\|\phi_0(D-\eta)r(D)R_0(z)f\|_{B^*}\leq C\|\phi_0(D-\eta)f\|_B,\quad\eta\in\mathbb{R}^n.
	\end{equation}
	Squaring both sides of \eqref{eq329} and integrating on $\mathbb{R}^n$ with respect to $\eta$, \eqref{eq324} when $z\in\mathbb{C}^\pm$ and $|z-z_0|<\epsilon$ is then a consequence of \eqref{eq22} and \eqref{eq23}.
	
	Now we can apply a finite covering to $K\cap\mathbb{R}$ in a $\mathbb{C}^+$- or $\mathbb{C}^-$-neighborhood, and obtain \eqref{eq324} when $|\mathrm{Im}\,z|<M$ for some $M>0$. The rest when $|\mathrm{Im}\,z|\geq M$ is a result of the trivial estimate
	\begin{equation*}
	\begin{split}
	\|r(D)R_0(z)f\|_{L^2}\leq&C\left\|\frac{\langle\cdot\rangle^s\hat{f}}{|\cdot|^s-\mathrm{Re}\,z-i\mathrm{Im}\,z}\right\|_{L^2}\\
	\leq&C_{\mathrm{Re}\,K,M}\|f\|_{L^2},
	\end{split}
	\end{equation*}
	and embeddings $B\hookrightarrow L^2\hookrightarrow B^*$. The proof of \eqref{eq324} is now complete.
\end{proof}

\section{Free Resolvent Estimates: II}\label{sec6}

The main result of this section is Theorem \ref{thm64}, and a use of Littlewood-Paley theory in Lemma \ref{lm62} is crucial. We shall establish weighted versions of results in Section \ref{sec5} for the free resolvents $R_0(\lambda\pm i0)$ acting on $f\in B$ with $\hat{f}=0$ on $\{p(\xi)=\lambda\}$, where such trace is actually well defined by the following lemma as mentioned in the Introduction.

\begin{lemma}[{\cite[Theorem 14.1.1]{Hor2}}]\label{lm23}
Let $M\subset\mathbb{R}^n$ be a compact $C^1$ hypersurface. Then the map
\begin{equation}
\mathscr{S}\ni u\mapsto\hat{u}|_M\in L^2(M,dS)
\end{equation}
can be extended to a surjection from $B$ to $L^2(M,dS)$ by continuity, where $dS$ is the Euclidean surface measure in $M$.
\end{lemma}

First comes the weighted version of Lemma \ref{lm31}.

\begin{lemma}[{\cite[Theorem 14.2.4]{Hor2}}]\label{lm61}
Let $N\in\mathbb{N}_+$, $X\subset\mathbb{R}^n$ be bounded and open, $\phi\in C_c^\infty(X)$, $C_1,C_2>0$, and given function $j:X\rightarrow\{1,\cdots,n\}$. Consider all $p\in C^{2+N}(X;\mathbb{R})$ such that
\begin{equation}
|\partial_{j(\xi)}p(\xi)|\geq C_1\,\,\text{and}\,\,|\partial^\alpha p(\xi)|\leq C_2,\quad\xi\in\mathrm{supp}\,\phi,\,\,|\alpha|\leq2+N.
\end{equation}
If $\lambda\in\mathbb{R}$, $f\in B$ and $\phi\hat{f}=0$ on $\{\xi\in\mathbb{R}^n;~p(\xi)=\lambda\}$, then $u_\lambda=\mathscr{F}^{-1}((p-\lambda\pm i0)^{-1}\phi\hat{f})$ is independent of the sign, and
\begin{equation}
\|\tilde{\mu}u_\lambda\|_{B^*}\leq C\|\tilde{\mu}f\|_B,
\end{equation}
where $\tilde{\mu}(\cdot)=\mu(|\cdot|)$, $\mu$ is any $C^1$ non-decreasing function satisfying
\begin{equation}
(1+t)\mu'(t)\leq N\mu(t),\quad t>0,
\end{equation}
and $C>0$ is independent of $p,\lambda,\mu$ and $f$.	
\end{lemma}

Before introducing the analogue of Lemma \ref{lm32}, we recall the Littlewood-Paley characterization of inhomogeneous Lipschitz space $\Lambda_\gamma(\mathbb{R}^n)$ where $\gamma\in\mathbb{R}_+$ (see \cite[Section 1.4]{gra2}). Let $\Phi\in C_c^\infty$ be non-negative, supported in $\{\xi\in\mathbb{R}^n;~\frac67\leq|\xi|\leq2\}$, and equal to $1$ on $\{\xi\in\mathbb{R}^n;~1\leq|\xi|\leq\frac{12}{7}\}$, satisfying
\begin{equation}
\sum_{j=-\infty}^{+\infty}\Phi(2^{-j}\xi)=1,\quad\xi\in\mathbb{R}^n\setminus\{0\}.
\end{equation}
Also let
\begin{equation}
\tilde{\Phi}(\xi)=\begin{cases}
\sum_{j=-\infty}^{0}\Phi(2^{-j}\xi),\quad&\xi\neq0,\\
1,&\xi=0.
\end{cases}
\end{equation}
We define $\Delta_jf=\mathscr{F}^{-1}(\Phi(2^{-j}\cdot)\hat{f})$ and $\tilde{\Delta}f=\mathscr{F}^{-1}(\tilde{\Phi}\hat{f})$ for $f\in\mathscr{S}'$. Then
\begin{equation}\label{eq66}
\|\tilde{\Delta}f\|_{L^\infty}+\sup_{j\geq1}2^{j\gamma}\|\Delta_jf\|_{L^\infty}\leq C\|f\|_{\Lambda_\gamma},\quad f\in\Lambda_\gamma(\mathbb{R}^n),
\end{equation}
where $\Lambda_\gamma(\mathbb{R}^n)$ consists of all continuous functions $f$ satisfying
\begin{equation}
\|f\|_{\Lambda_\gamma}=\|f\|_{L^\infty}+\sup_{x\in\mathbb{R}^n}\sup_{h\in\mathbb{R}^n\setminus\{0\}}\frac{|D_h^{[\gamma]+1}(f)(x)|}{|h|^\gamma}<\infty,
\end{equation}
and here $D_h$ is the difference operator at length $h$.

\begin{lemma}\label{lm62}
Let $s>0$, $N\in\mathbb{N_+}$, $X\subset\mathbb{R}^n$ be bounded and open, $\phi\in C_c^\infty(X)$, and $Q_0\in C(X;\mathbb{R})$. For some $\epsilon>0$, let $K_1,K_2$ be defined in Lemma \ref{lm32}, and $\phi_2$ be defined in \eqref{eq36}. We assume that $Q_0\in C^{2+N}(K_1)$ with $\nabla Q_0\neq0$ on $K_1$, and $\phi_2Q_0\in\Lambda_s$. Then there exist $\delta_0,C>0$, such that when $Q\in C(X;\mathbb{R})\cap C^{2+N}(K_1)$ satisfies
\begin{equation}
\|Q-Q_0\|_{C^{2+N}(K_1)}+\|Q-Q_0\|_{C(K_2)}<\delta_0,\quad\|\phi_2Q\|_{\Lambda_s}<\delta_0^{-1},
\end{equation}
and when $\lambda\in(-\delta_0,\delta_0)$, $f\in B$ with $\phi\hat{f}=0$ on $\{\xi\in\mathbb{R}^n;~Q(\xi)=\lambda\}$, we have
\begin{equation}
\|\tilde{\mu}\mathscr{F}^{-1}((Q-\lambda)^{-1}\phi\hat{f})\|_{B^*}\leq C\|\tilde{\mu}f\|_B,
\end{equation}
where $\tilde{\mu}(\cdot)=\mu(|\cdot|)$ and $\mu$ is any $C^1$ non-decreasing function satisfying
\begin{equation}\label{eq610}
(1+t)\mu'(t)\leq\min\{\mbox{$s+\frac12,N$}\}\mu(t),\quad t>0.
\end{equation}
\end{lemma}

We note that by Lemma \ref{lm61},
\begin{equation*}
\mathscr{F}^{-1}((Q_0-\lambda)^{-1}\phi\hat{f})=\mathscr{F}^{-1}((Q_0-\lambda+ i0)^{-1}\phi\hat{f})=\mathscr{F}^{-1}((Q_0-\lambda- i0)^{-1}\phi\hat{f})\in\mathscr{S}'
\end{equation*}
is well defined, and of course we will show that this is true for $Q$ if $\delta_0$ is small by the same reason.

\begin{proof}[Proof of Lemma \ref{lm62}]
As in the proof of Lemma \ref{lm32}, with $\phi=\phi_1+\phi_2$, we first use Lemma \ref{lm61} instead of Lemma \ref{lm31} to obtain for sufficiently small $\delta_0$ that
\begin{equation*}
\|\tilde{\mu}\mathscr{F}^{-1}((Q-\lambda)^{-1}\phi_1\hat{f})\|_{B^*}\leq C\|\tilde{\mu}f\|_B.
\end{equation*}
On the other hand, when $|\lambda|<\delta_0\leq\epsilon$ and $\xi\in\mathrm{supp}\,\phi_2\subset K_2$, we have $|Q(\xi)-\lambda|\geq\epsilon$, and therefore
\begin{equation}\label{eq612}
\|(Q-\lambda)^{-1}\phi_2\|_{\Lambda_s}\leq C(\delta_0),\quad\text{if}\,\,\|Q-Q_0\|_{C(K_2)}<\delta_0\,\,\text{and}\,\,\|\phi_2Q\|_{\Lambda_s}<\delta_0^{-1}.
\end{equation}
Now we are left to show
\begin{equation}\label{eq613}
\|\tilde{\mu}u\|_{B^*}\leq C\|\tilde{\mu}f\|_B,
\end{equation}
where $u=\mathscr{F}^{-1}((Q-\lambda)^{-1}\phi_2\hat{f})$, by using \eqref{eq612}. 

We first have
\begin{equation}\label{eq614}
\begin{split}
\|\tilde{\mu}u\|_{B^*}\leq&\sup_{j\geq1}R_j^{-\frac12}\mu(R_j)\|u\|_{L^2(X_j)}\\
\leq&\sup_{j\geq1}R_j^{-\frac12}\mu(R_j)\sum_{k\geq1}\left\|\int_{X_k}\mathscr{F}^{-1}((Q-\lambda)^{-1}\phi_2)(x-y)f(y)dy\right\|_{L_x^2(X_j)}.
\end{split}
\end{equation}
If $x\in X_j$ and $y\in X_k$, we have
\begin{equation*}
x-y\in\begin{cases}
\tilde{X}_{\max\{j,k\}},\quad&|j-k|\geq2,\\
\overline{X_0\cup\cdots\cup X_{j+1}},&|j-k|\leq1,
\end{cases}
\end{equation*}
and thus
\begin{equation}\label{eq616}
\begin{split}
\|u\|_{L^2(X_j)}\leq&\sum_{|k-j|\geq2}\left\|\left(\sum_{l=\max\{j,k\}-2}^{\max\{j,k\}+1}\Phi(2^{-l}\cdot)\mathscr{F}^{-1}((Q-\lambda)^{-1}\phi_2)\right)\ast\chi_kf\right\|_{L^2(\mathbb{R}^n)}\\
&\quad+\sum_{|k-j|\leq1}\left\|\left(\sum_{l=-\infty}^{j+2}\Phi(2^{-l}\cdot)\mathscr{F}^{-1}((Q-\lambda)^{-1}\phi_2)\right)\ast\chi_kf\right\|_{L^2(\mathbb{R}^n)}\\
\leq&C\sum_{|k-j|\geq2}\left(\sum_{l=\max\{j,k\}-2}^{\max\{j,k\}+1}\|\Delta_l((Q-\lambda)^{-1}\phi_2)\|_{L^\infty}\right)\|f\|_{L^2(X_k)}\\
&\quad+C\sum_{|k-j|\leq1}\left(\|\tilde{\Delta}((Q-\lambda)^{-1}\phi_2)\|_{L^\infty}+\sum_{l=1}^{j+2}\|\Delta_l((Q-\lambda)^{-1}\phi_2)\|_{L^\infty}\right)\|f\|_{L^2(X_k)}.
\end{split}
\end{equation}

We shall use \eqref{eq66} and \eqref{eq612} to treat \eqref{eq616}, before which we may assume without loss of generality that $\mu(0)=1$, then \eqref{eq610} implies
\begin{equation*}
\mu(R_j)\leq2^{s+\frac12}\mu(R_{j-1})\leq CR_j^{s+\frac12}.
\end{equation*}
First,
\begin{equation}\label{eq618}
\begin{split}
&R_j^{-\frac12}\mu(R_j)\sum_{|k-j|\geq2}\left(\sum_{l=\max\{j,k\}-2}^{\max\{j,k\}+1}\|\Delta_l((Q-\lambda)^{-1}\phi_2)\|_{L^\infty}\right)\|f\|_{L^2(X_k)}\\
\leq&CR_j^{-\frac12}R_j^{s+\frac12}\sum_{|k-j|\geq2}R_{\max\{j,k\}}^{-s}\|f\|_{L^2(X_k)}\\
\leq&C\|f\|_B\\
\leq&C\|\tilde{\mu}f\|_B.
\end{split}
\end{equation}
We also have
\begin{equation}\label{eq619}
\begin{split}
&R_j^{-\frac12}\mu(R_j)\sum_{|k-j|\leq1}\left(\|\tilde{\Delta}((Q-\lambda)^{-1}\phi_2)\|_{L^\infty}+\sum_{l=1}^{j+2}\|\Delta_l((Q-\lambda)^{-1}\phi_2)\|_{L^\infty}\right)\|f\|_{L^2(X_k)}\\
\leq&C\mu(R_j)\sum_{|k-j|\leq1}(1+2^{-s}+\cdots+2^{-(j+2)s})\|f\|_{L^2(X_k)}\\
\leq&C\mu(R_j)\sum_{|k-j|\leq1}\mu(R_{k-1})^{-1}\|\tilde{\mu}f\|_{L^2(X_k)}\\
\leq&C\|\tilde{\mu}f\|_B.
\end{split}
\end{equation}
Now \eqref{eq614}, \eqref{eq616}, \eqref{eq618} and \eqref{eq619} imply \eqref{eq613}, which completes the proof.
\end{proof}

Notice that the function $|\xi|^s$ locally belongs to $\Lambda_s$ for $s>0$, Lemma \ref{lm62} is then applicable in the following lemma.

\begin{lemma}\label{lm63}
Given $s>0$, $\lambda_0\in\mathbb{R}\setminus\{0\}$ and $\phi\in C_c^\infty(B(1))$. For $\eta\in\mathbb{R}^n$, we define
\begin{equation}
Q_\eta(\xi)=\frac{|\xi+\eta|^s-\lambda_0}{\langle\eta\rangle^s},\quad\xi\in\mathbb{R}^n.
\end{equation}
Then there exist $\delta_0,C>0$ such that when $\lambda\in(-\delta_0,\delta_0)$, $\eta\in\mathbb{R}^n$, and $f\in B$ with $\phi\hat{f}=0$ on $\{\xi\in\mathbb{R}^n;~Q_\eta(\xi)=\lambda\}$, we have
\begin{equation}
\|\tilde{\mu}\mathscr{F}^{-1}((Q_\eta-\lambda)^{-1}\phi\hat{f})\|_{B^*}\leq C\|\tilde{\mu}f\|_B,
\end{equation}
where $\tilde{\mu}(\cdot)=\mu(|\cdot|)$ and $\mu$ is any $C^1$ non-decreasing function satisfying
\begin{equation}\label{equ523}
(1+t)\mu'(t)\leq(\mbox{$s+\frac12$})\mu(t),\quad t>0.
\end{equation}
\end{lemma}

\begin{proof}
The proof is almost the same to that of Lemma \ref{lm33}, if we use Lemma \ref{lm62} (with $N=[s+\frac12]+1$) instead of Lemma \ref{lm32} to show that $\mathcal{A}$ is precompact under the $C^{[s+\frac12]+1}$ topology; and $\|\phi Q_\eta\|_{\Lambda_s}$ is bounded uniformly in $\{\eta\in\mathbb{R}^n;~|\eta|\leq2\}$, if 
\begin{equation*}
\phi(\xi)=(1-h(Q_{\eta_0}(\xi)))\phi(\xi)
\end{equation*}
for any fixed $|\eta_0|\leq2$.
\end{proof}

Now we prove the main result of this section.

\begin{theorem}\label{thm64}
Let $s>0$, $I\subset\mathbb{R}\setminus\{0\}$ be compact, and suppose $r\in C^{[s]+2}(\mathbb{R}^n)$ with
\begin{equation}
\sum_{|\alpha|\leq[s]+2}\sup_{\xi\in\mathbb{R}^n}\left|\partial_\xi\frac{r(\xi)}{\langle\xi\rangle^s}\right|<\infty.
\end{equation}
Then there exists $C_I>0$ such that when $\lambda\in I$ and when $f\in B$ with $\hat{f}=0$ on $\{\xi\in\mathbb{R}^n;~|\xi|^s=\lambda\}$, we have
\begin{equation}
\|\tilde{\mu}r(D)R_0(\lambda\pm i0)f\|_{B^*}\leq C_I\|\tilde{\mu}f\|_B,
\end{equation}
where $\tilde{\mu}(\cdot)=\mu(|\cdot|)$ and $\mu$ is any $C^1$ non-decreasing function satisfying
\begin{equation}\label{equ527}
(1+t)\mu'(t)\leq(\mbox{$s+\frac12$})\mu(t),\quad t>0.
\end{equation}
\end{theorem}

\begin{proof}
As in the proof of Theorem \ref{thm34}, we use Lemma \ref{lm63} instead of Lemma \ref{lm33} to obtain for all $\eta\in\mathbb{R}^n$ that
\begin{equation*}
\langle\eta\rangle^s\|\tilde{\mu}\phi_0(D-\eta)r(D)R_0(\lambda\pm i0)f\|_{B^*}\leq C\|\tilde{\mu}r(D)\phi_0(D-\eta)f\|_B.
\end{equation*}
For the rest, we use analogues of Lemma \ref{lm21} and Lemma \ref{lm22} associated with the weight $\tilde{\mu}$, which requires the smoothness index $[s]+2$ of $r$, to complete the proof. We refer to Theorem 14.1.7 and Corollary 14.1.8 in \cite{Hor2} for such analogues.
\end{proof}

\begin{remark}
When $s=2m$ where $m\in\mathbb{N}_+$, the function $|\xi|^s$ is smooth and thus locally belongs to $\Lambda_N$ for any $N>0$. In such case, $s+\frac12$ in the right hand side of \eqref{equ523} and of \eqref{equ527} can be replaced by any positive number. This could finally result in \eqref{equ113} that we mentioned earlier, but we neglect such discussion.
\end{remark}

\section{Discreteness of Partial Eigenvalues}\label{sec7}

This section is devoted to studying a special part $\Lambda$ of $\sigma_\mathrm{pp}$, whose discreteness determines the proof of asymptotic completeness of $W_\pm$ in the next section, and the main result is Proposition \ref{pro72}. We will end this section by showing
\begin{equation}\label{61}
\Lambda\subset\sigma_\mathrm{pp}\setminus\{0\}.
\end{equation}
However, we will show at the end of Section \ref{sec8} that they are actually equal.

For $\lambda\in\mathbb{R}\setminus\{0\}$, the strategy to study the equation $(H-\lambda)u=0$ in $L^2$ is to first consider a wider sense of solutions:
\begin{equation}\label{equ61}
(-\Delta)^\frac s2u-\lambda u=-Vu,\quad u\in B_s^*.
\end{equation}
Notice that when $s$ is not an even number, $(-\Delta)^\frac s2$ does not act on all temperate distributions for the non-smoothness of its Fourier symbol, thus we actually interpret equation \eqref{equ61} in the sense that
\begin{equation}\label{equ62}
\lim_{\epsilon\rightarrow0}((-\Delta)^\frac s2-\lambda)\phi_\epsilon(D)u=f\in B\quad\text{in}~\mathscr{S}',
\end{equation}
and $f=-Vu$, where $\hat{\phi}_\epsilon$ is a smooth Fourier cutoff function avoiding an $\epsilon$-ball centered at $0$. In order to easier verify such limit, we are interested in the following definition.

\begin{definition}\label{def71}
Let $s>0$ and V be a short range potential. We define $\Lambda$ to be the set of $\lambda\in\mathbb{R}\setminus\{0\}$ such that the equation
\begin{equation}\label{eq72}
(I+VR_0(\lambda+i0))f=0
\end{equation}
has a non-trivial solution $f\in B$.
\end{definition}

\begin{remark}\label{rk62}
If $\lambda\in\Lambda$, $0\neq f\in B$ solves \eqref{eq72}, and set $u=R_0(\lambda+i0)f\in B_s^*$ by Theorem \ref{thm34}, we have $f=-Vu$ and that
\begin{equation}\label{e72}
u=-R_0(\lambda+i0)Vu.
\end{equation}
Therefore, \eqref{equ62} even holds in the weak* topology of $B^*$ by dominated convergence.
\end{remark}

We will first focus on the study of \eqref{e72} rather than \eqref{equ62}. Let $p(\xi)=|\xi|^s$ and $\lambda\in\mathbb{R}\setminus\{0\}$. Recall that the formula
\begin{equation}\label{eq73}
R_0(\lambda+i0)f-R_0(\lambda-i0)f=2\pi i\mathscr{F}^{-1}(\delta(p-\lambda)\hat{f})
\end{equation}
holds in $\mathscr{S}'$ if $\hat{f}\in C_c^\infty$, for $\lambda$ is not a critical point of $p$, which is smooth near $\mathrm{supp}\,\delta(p-\lambda)$. \eqref{eq73} also holds when $f\in B$ while both sides are in $B_s^*$ by Theorem \ref{thm34}, and on the right hand side the fact that $\delta(p-\lambda)\hat{f}$ can be interpreted as the $L^2$ trace on the surface comes from Lemma \ref{lm23} and an obvious duality discussion. We first introduce a lemma for the inhomogeneous equation.

\begin{lemma}\label{lm72}
Suppose $\lambda\in\mathbb{R}\setminus\{0\}$, and $u\in B^*$ satisfies $((-\Delta)^\frac s2-\lambda)u=f\in B$ in the sense of \eqref{equ62}. Then $\mathscr{F}\left(u-R_0(\lambda\pm i0)f\right)$ are supported in $\{\xi\in\mathbb{R}^n;~|\xi|^s=\lambda\}$ as $L^2$ densities. Moreover, if $\mathrm{Im}\langle u,f\rangle=0$, then each of $u=R_0(\lambda+i0)f$ and $u=R_0(\lambda-i0)f$ implies each other.
\end{lemma}
\begin{proof}
See \cite[Theorem 14.3.8]{Hor2}. We note that the reference is for differential operators, but its proof works well here under Theorem \ref{thm34} and the comment after \eqref{eq73}. 
\end{proof}

\begin{remark}\label{rk74}
If $\lambda\in\Lambda$ and $f\in B$ solves \eqref{eq72}, since $V$ is real-valued, Lemma \ref{lm72} and Remark \ref{rk62} then imply that $f$ also solves
\begin{equation}
(I+VR_0(\lambda-i0))f=0.
\end{equation}
\end{remark}

The main properties of $\Lambda$ are indicated in the following.

\begin{proposition}\label{pro72}
Suppose $\lambda\in\mathbb{R}\setminus\{0\}$, and $V$ is a short range potential. Consider $u\in B_s^*$ solving
\begin{equation}\label{eq718}
u=-R_0(\lambda+i0)Vu.
\end{equation}
Then\\
(i) The space of solutions to \eqref{eq718} is finite dimensional.\\
(ii) There exists $C>0$ independent of $u$ such that
\begin{equation}\label{eq719}
\|\langle\cdot\rangle^{s+\frac12}J_su\|_{B^*}\leq C\|Vu\|_B<\infty,
\end{equation}
where $C$ remains bounded if $\lambda$ stays in a bounded set. Consequently,
\begin{equation}\label{equ69}
\|\langle\cdot\rangle^{s-\epsilon}J_{s'}u\|_{L^2}\leq C_\epsilon\|Vu\|_B<\infty,
\end{equation}
for any $\epsilon>0$ and $s'\leq s$, where $C_\epsilon$ remains bounded if $\lambda$ stays in a bounded set. Thus $u\in H^s$ and $(H-\lambda)u=0$ holds.\\
(iii) The set of $\lambda\in\mathbb{R}\setminus\{0\}$ such that \eqref{eq718} has a non-trivial $B_s^*$ solution is discrete in $\mathbb{R}\setminus\{0\}$.\\
\end{proposition}

\begin{proof}
(i) is obvious since $R_0(\lambda+i0)V$ is compact in $B_s^*$.

For (ii), let $f=-Vu\in B$. Since $V$ is real-valued, by Lemma \ref{lm72} we have $u=R_0(\lambda+i0)f=R_0(\lambda-i0)f$, and thus that $\hat{f}=0$ on $\{\xi\in\mathbb{R}^n;~|\xi|^s=\lambda\}$ in the view of \eqref{eq73}. Set $\mu_\epsilon(t)=(1+t)^{s+\frac12}(1+\epsilon t)^{-s-\frac12}$ for $\epsilon\in(0,1)$, one checks that $\mu_\epsilon$ is an increasing function satisfying
\begin{equation}\label{eq720}
0<(1+t)\mu_\epsilon'(t)<\mbox{$(s+\frac12)$}\mu_\epsilon(t),\quad t>0.
\end{equation}
Now we apply Theorem \ref{thm64} with $r(D)=J_s$ and such weight function $\mu_\epsilon$ to the above $f$ to obtain
\begin{equation}\label{eq721}
\begin{split}
\|\tilde{\mu}_\epsilon J_su\|_{B^*}=\|\tilde{\mu}_\epsilon J_sR_0(\lambda+i0)f\|_{B^*}\leq C\|\tilde{\mu}_\epsilon Vu\|_B\leq C\sum_{j\geq1}\mu_\epsilon(R_j)R_j^\frac12\|Vu\|_{L^2(X_j)},
\end{split}
\end{equation}
where $\tilde{\mu}_\epsilon(\cdot)=\mu_\epsilon(|\cdot|)$, and $C>0$ remains bounded if $\lambda$ stays in a bounded set. We note that $\mu_\epsilon$ is a bounded function, thus the left hand side of \eqref{eq721} is finite, which makes absorption discussion reasonable. By \eqref{eq411}, 
\begin{equation}\label{eq722}
\mu_\epsilon(R_j)R_j^\frac12\|Vu\|_{L^2(X_j)}\leq CM_jR_j^\frac12\mu_\epsilon(R_j)\left(\|J_su\|_{L^2(\tilde{X}_j)}+e^{-c_2R_j}\|e^{-c_1|\cdot|}J_su\|_{L^1(\tilde{X}_j^c)}\right).
\end{equation}
Notice that \eqref{eq720} and $\mu_\epsilon(0)=1$ imply $\mu_\epsilon(R_j)\leq 2^{s+\frac12}\mu_\epsilon(R_{j-1})\leq CR_j^{s+\frac12}$, we first have
\begin{equation}\label{eq723}
\|J_su\|_{L^2(\tilde{X}_j)}\leq C\mu_\epsilon^{-1}(R_{j-2})\|\tilde{\mu}_\epsilon J_su\|_{L^2(\tilde{X}_j)}\leq C\mu_\epsilon^{-1}(R_j)R_j^\frac12\|\tilde{\mu}_\epsilon J_su\|_{B^*}.
\end{equation}
We also have
\begin{equation}\label{eq724}
\|e^{-c_1|\cdot|}J_su\|_{L^1(\tilde{X}_j^c)}\leq C\sum_{l\geq1}e^{-c_1R_{l-1}}R_l^\frac{n+1}{2}R_l^{-\frac12}\|\tilde{\mu}_\epsilon J_su\|_{L^2(X_l)}\leq C\|\tilde{\mu}_\epsilon J_su\|_{B^*}.
\end{equation}
Combining \eqref{eq720}-\eqref{eq724}, we obtain for any integer $J$ that
\begin{equation}\label{eq725}
\|\tilde{\mu}_\epsilon J_su\|_{B^*}\leq C\sum_{j=1}^J\mu_\epsilon(R_j)R_j^\frac12\|Vu\|_{L^2(X_j)}+C\sum_{j>J}\left(M_jR_j+M_jR_j^{1+s}e^{-c_2R_j}\right)\|\tilde{\mu}_\epsilon J_su\|_{B^*}.
\end{equation}
By \eqref{eq412}, we can fixed large $J$ such that the second summation of \eqref{eq725} is absorbed into the left hand side. We note that such choice of $J$ only depends on $M_j$ and $\lambda$, and the choice is uniform if $\lambda$ stays in a bounded set. We therefore obtain
\begin{equation*}
\|\tilde{\mu}_\epsilon J_su\|_{B^*}\leq C\mu_\epsilon(R_J)\|Vu\|_B.
\end{equation*}
Sending $\epsilon$ to $0$ proves \eqref{eq719}, for $(1+|\cdot|)^{s+\frac12}$ is equivalent to $\langle\cdot\rangle^{s+\frac12}$. It is easy to see that
\begin{equation*}
\|\langle\cdot\rangle^{s-\epsilon}J_su\|_{L^2}\leq C\sum_{j\geq1}R_j^{-\epsilon}\|\langle\cdot\rangle^{s+\frac12}J_su\|_{B^*}\leq C_\epsilon\|\langle\cdot\rangle^{s+\frac12}J_su\|_{B^*},
\end{equation*}
\eqref{equ69} is then proved, if we observe from \eqref{eq41} that $\langle\cdot\rangle^{s-\epsilon}J_{-(s-s')}\langle\cdot\rangle^{-(s-\epsilon)}$ is $L^2$ bounded, for its integral kernel is dominated by the convolution kernel of a Bessel potential operator near the diagonal, and decays exponentially away from the diagonal. By Remark \ref{rk62}, the proof of (ii) is complete.

For (iii), suppose we have $\mathbb{R}\setminus\{0\}\ni\lambda_j\rightarrow\lambda\in\mathbb{R}\setminus\{0\}$ where $\lambda_j$'s are distinct, and $u_j\in B_s^*$ with $\|u_j\|_{B_s^*}=1$ satisfying
\begin{equation}\label{eq727}
u_j=-R_0(\lambda_j+i0)Vu_j.
\end{equation}
Since \eqref{equ69} implies $\sup_j\|u_j\|_{L^2}\leq C\|V\|_{B_s^*\rightarrow B}$, 
\begin{equation*}
\lim_{R\rightarrow\infty}\sup_j\int_{|\xi|>R}|\hat{u}_j|^2d\xi=0,\quad\text{and}\quad\lim_{R\rightarrow\infty}\sup_j\int_{|x|>R}|u_j|^2dx=0,
\end{equation*}
the Kolmogorov-Riesz theorem (see e.g. \cite{P}) indicates that we can assume $u_j\rightarrow u$ in $L^2$ up to subsequence. On the other hand, \eqref{eq727} and the fact that $u_j\in\mathrm{Dom}(H)=H^s$ imply $Hu_j=\lambda_ju_j$, and since $H$ is closed, we also have $Hu=\lambda u$. We must have $u=0$ for it is orthogonal to a bounded sequence of orthogonal eigenfunctions. Now Lemma \ref{lm76} implies $\|Vu_j\|_B\rightarrow0$, we thus conclude by \eqref{eq324} and the assumption $\lambda\in\mathbb{R}\setminus\{0\}$ that
\begin{equation*}
\|u_j\|_{B_s^*}=\|R_0(\lambda_j+i0)Vu_j\|_{B_s^*}\rightarrow0,
\end{equation*}
which is a contradiction. The proof of (iii) is complete.
\end{proof}

\begin{remark}\label{rk65}
Now \eqref{61} apparently holds for $\Lambda$ is a part of the set of $\lambda$ that Proposition \ref{pro72} considers, which turn out to be genuine eigenvalues of $H$.
\end{remark}

\section{Asymptotic Completeness of $W_\pm$}\label{sec8}

Let $\{E_\lambda\}$ be the spectral family of $H$, and
\begin{equation}\label{eq87}
\tilde{\Lambda}=\Lambda\cup\{0\},\quad E^c=\int_{\mathbb{R}\setminus\tilde{\Lambda}}dE_\lambda,\quad E^d=\int_{\tilde{\Lambda}}dE_\lambda.
\end{equation}
The goal of this section is to prove Theorem \ref{thm14}, and the proof is in principle given in the Section 14.6 of H\"{o}rmander \cite{Hor2}. The plan is to consider the distorted Fourier transforms formally given by $F_\pm=\mathscr{F}(W_\pm)^*$, and show that $F_\pm$ are both scaled isometries from $E^cL^2$ to $L^2$. However, unlike many authors who defined the distorted Fourier transforms by constructing generalized eigenfunction expansion, H\"{o}rmander \cite{Hor2} considers $F_\pm$ constructed on lower dimensional hypersurfaces using spectral calculus. In particular, such construction is tricky when considering the trace property Lemma \ref{lm23} in a parameterized way. We shall supplement some details for such construction, especially in the proofs of Lemma \ref{lm82} and Theorem \ref{thm74}. We note that such purpose is based on the discreteness of $\Lambda$ proved in Proposition \ref{pro72}.

If $z\in\mathbb{C}\setminus\mathbb{R}$, let $R(z)=(H-z)^{-1}$ denote the $L^2$ resolvent of $H$. We start by checking its boundary behavior when $\mathrm{Im}\,z\rightarrow\pm 0$ using the claims for $R_0(z)$ that we have established. 

\begin{lemma}\label{lm81}
If $z\in\mathbb{C}^\pm\setminus\tilde{\Lambda}$ and $V$ is a short range potential, then $(I+VR_0(z))^{-1}$ exists and is continuous in $B$. Moreover, the maps
\begin{equation}\label{eq81}
\mathbb{C}^\pm\setminus\tilde{\Lambda}\ni z\mapsto(I+VR_0(z))^{-1}f\in B
\end{equation}
are continuous if $f\in B$.
\end{lemma}

\begin{proof}
When $\mathrm{Im}\,z\neq0$, for $u\in B$ we have $R_0(z)u\in H^s$, and it is easy to deduce that
\begin{equation}\label{eq83}
R_0(z)u=R(z)(I+VR_0(z))u,
\end{equation}
where $I+VR_0(z)$ can be interpreted as an operator in $B$. Thus $I+VR_0(z)$ has kernel $\{0\}$, which is also true when $z\in\mathbb{C}^\pm\setminus\tilde{\Lambda}$ by Definition \ref{def71} and Remark \ref{rk74}. By Fredholm theory (see \cite[Lemma 14.5.3]{Hor2}), $(I+VR_0(z))^{-1}$ is bounded in $B$ and strongly continuous in $z\in\mathbb{C}^\pm\setminus\tilde{\Lambda}$, where the proof needs the facts that $\mathbb{C}^\pm\setminus\{0\}\ni z\mapsto VR_0(z)f\in B$ is continuous when $f\in B$ (see Theorem \ref{thm34} and Lemma \ref{lm76}), and $\{VR_0(z)f;~\|f\|_B\leq1,\,z\in K\}$ is precompact in $B$ if $K$ is bounded, which obviously holds.
\end{proof}

If we let $u=(I+VR_0(z))^{-1}f$ in \eqref{eq83}, Lemma \ref{lm81} shows that 
\begin{equation*}
R(z)f=R_0(z)(I+VR_0(z))^{-1}f,\quad f\in B,
\end{equation*}
is well defined for $z\in\mathbb{C}^\pm\setminus\tilde{\Lambda}$, and the maps $\mathbb{C}^\pm\setminus\tilde{\Lambda}\ni z\mapsto R(z)f\in B^*$ are weak* continuous. In particular
\begin{equation*}
R(\lambda\pm i0)f=R_0(\lambda\pm i0)(I+VR_0(\lambda\pm i0))^{-1}f\in B_s^*,\quad f\in B,~\lambda\in\mathbb{R}\setminus\tilde{\Lambda},
\end{equation*}
and we write $f_{\lambda\pm i0}=(I+VR_0(\lambda\pm i0))^{-1}f$ for later convenience.

Let $M_\lambda=\{\xi\in\mathbb{R}^n;~|\xi|^s=\lambda\}$. (Note that $M_\lambda=\emptyset$ if $\lambda<0$.) The spectral theory (see \cite[p. 255-256]{Hor2}) implies that
\begin{equation}\label{eq88}
\begin{split}
\|E^cf\|_{L^2}=&(2\pi)^{-n}\int_{\mathbb{R}\setminus\tilde{\Lambda}}\frac{d\lambda}{s\lambda^\frac{s-1}{s}}\int_{M_\lambda}\left|\hat{f}_{\lambda+i0}(\xi)\right|^2dS\\
=&(2\pi)^{-n}\int_{\mathbb{R}\setminus\tilde{\Lambda}}\frac{d\lambda}{s\lambda^\frac{s-1}{s}}\int_{M_\lambda}\left|\hat{f}_{\lambda-i0}(\xi)\right|^2dS,\quad f\in B.
\end{split}
\end{equation}
Note that $d\xi=d\lambda dS/(s\lambda^\frac{s-1}{s})$ where $\xi\in\mathbb{R}^n$, the above repeated integral is formally the integration of $\hat{f}_{|\xi|^s\pm i0}(\xi)$ in $\mathbb{R}^n$. However by Lemma \ref{lm23}, the $L^2$ trace $\hat{f}_{\lambda\pm i0}(\xi)$ is only defined for almost every $\xi\in M_\lambda$ with respect to the surface measure by extension. We therefore need to choose the value of $\hat{f}_{|\xi|^s\pm i0}(\xi)$ well on each $M_\lambda$ even just to make it measurable in $\mathbb{R}^n$.

\begin{lemma}\label{lm82}
If $f\in B$, then there exist measurable functions $F_\pm f$ which are uniquely defined almost everywhere in $\mathbb{R}^n$ with respect to the Lebesgue measure, such that for every $\lambda\in\mathbb{R}\setminus\tilde{\Lambda}$ we have $F_\pm f(\xi)=\hat{f}_{\lambda\pm i0}(\xi)$ holds for almost every $\xi\in M_\lambda$ with respect to the surface measure.
\end{lemma}

\begin{proof}
By Proposition \ref{pro72}, first notice that $\cup_{\lambda\in\tilde{\Lambda}}M_\lambda$ has measure $0$, the uniqueness follows. The discreteness of $\Lambda$ guarantees the existence of partition of unity
\begin{equation*}
\sum_{k=1}^{\infty}\psi_k(\lambda)=1,\quad \lambda\in\mathbb{R}\setminus\tilde{\Lambda},
\end{equation*}
where $\psi_k\in C_c^\infty(\mathbb{R})$ and the above summation is locally finite.

We first claim that for every $\lambda\in\mathbb{R}\setminus\tilde{\Lambda}$ and $j>0$, there exists $g_\lambda^j\in\mathscr{F}^{-1}C_c^\infty$ such that $\hat{g}_\lambda^j$ is continuous in $\lambda$ and
\begin{equation*}
\|f_{\lambda\pm i0}-g_\lambda^j\|_B<2^{-j}.
\end{equation*} 
For such purpose, we first fix $k$, and choose a partition of unity $\phi_{k,1},\cdots,\phi_{k,K}$ such that $\sum_{l=1}^K\phi_{k,l}=1$ in $\mathrm{supp}\,\psi_k$ where the summation is locally twice. By Lemma \ref{lm81}, if the partition is sufficiently fine, we have
\begin{equation}\label{eq810}
\|f_{\lambda\pm i0}-f_{\lambda'\pm i0}\|_B<\mbox{$\frac14$}\cdot2^{-k}2^{-j},\quad\lambda,\lambda'\in\mathrm{supp}\,\phi_{k,l},~l=1,\cdots,K.
\end{equation}
Take $\lambda_{k,l}\in\mathrm{supp}\,\phi_{k,l}$, and since $\mathscr{F}^{-1}C_c^\infty$ is dense in $B$, we can find $g_{k,l}^j\in\mathscr{F}^{-1}C_c^\infty$ such that
\begin{equation}\label{eq811}
\|f_{\lambda_{k,l}\pm i0}-g_{k,l}^j\|_B<\mbox{$\frac14$}\cdot2^{-k}2^{-j}.
\end{equation}
Set
\begin{equation*}
g_{k,\lambda}^j=\sum_{l=1}^K\phi_{k,l}(\lambda)g_{k,l}^j,\quad\lambda\in\mathrm{supp}\,\psi_k,
\end{equation*}
we have $\hat{g}_{k,\lambda}^j\in C_c^\infty$ continuous in $\lambda$ when $\lambda\in\mathrm{supp}\,\psi_k$, and 
\begin{equation*}
\|f_{\lambda\pm i0}-g_{k,\lambda}^j\|_B<2^{-k}2^{-j},\quad\lambda\in\mathrm{supp}\,\psi_k,
\end{equation*}
holds by \eqref{eq810} and \eqref{eq811}. Now let $g_\lambda^j=\sum_{k=1}^\infty\psi_k(\lambda)g_{k,\lambda}^j$, then $g_{\lambda}^j$ has the desired properties.

By Lemma \ref{lm23}, if $\lambda\in\mathbb{R}\setminus\tilde{\Lambda}$, we have
\begin{equation}\label{eq814}
\int_{M_\lambda}|\hat{f}_{\lambda\pm i0}-\hat{g}_\lambda^j|^2dS\leq C(\lambda)\|f_{\lambda\pm i0}-g_\lambda^j\|_B\leq C(\lambda)2^{-j},
\end{equation}
where $C(\lambda)$ is locally bounded by \eqref{eq324}. Thus for any compact $K\subset\mathbb{R}^n\setminus\cup_{\lambda\in\tilde{\Lambda}}M_\lambda$ and $k>0$, we have
\begin{equation}\label{eq815}
\int_K|\hat{g}_{|\xi|^s}^{j+k}(\xi)-\hat{g}_{|\xi|^s}^j(\xi)|^2d\xi\leq C_K2^{-j},
\end{equation}
which implies that $\hat{g}_{|\xi|^s}^j(\xi)$ has almost everywhere convergence in $\mathbb{R}^n$ up to subsequence, and we denote the almost everywhere limit to be $F_\pm f(\xi)$. We note that \eqref{eq815} implies $\hat{g}_{|\xi|^s}^j(\xi)\rightarrow F_\pm f(\xi)$ in $L_\mathrm{loc}^2(\mathbb{R}^n\setminus\cup_{\lambda\in\tilde{\Lambda}}M_\lambda)$ at the rate $2^{-j}$. 

By Fubini's theorem, 
\begin{equation*}
h(\lambda)=\int_{M_\lambda}|\hat{f}_{\lambda\pm i0}(\xi)-F_\pm f(\xi)|^2dS
\end{equation*}
exists for almost every $\lambda\in\mathbb{R}\setminus\tilde{\Lambda}$, and it follows from \eqref{eq814} that
\begin{equation*}
h(\lambda)\leq h_j(\lambda)=C(\lambda)2^{-j}+\int_{M_\lambda}|\hat{g}_\lambda^j(\xi)-F_\pm f(\xi)|^2dS.
\end{equation*}
Now we know from \eqref{eq815} for any compact $I\subset\mathbb{R}\setminus\tilde{\Lambda}$ that 
\begin{equation*}
\int_I h_j(\lambda)d\lambda\leq C_I2^{-j},
\end{equation*}
which implies up to subsequence that, $h_j(\lambda)\rightarrow0$ for almost every $\lambda\in\mathbb{R}\setminus\tilde{\Lambda}$, and thus $h(\lambda)=0$ for almost every $\lambda$. Therefore, we can modify $F_\pm f$ on a collection of $M_\lambda$ which form a set of measure $0$ in $\mathbb{R}^n$, and the proof is complete.
\end{proof}

$F_\pm f$ are called the distorted Fourier transforms of $f\in B$. From \eqref{eq88}, we know that $F_\pm$ has an extension supported onto $E^cL^2$.

\begin{theorem}\label{thm73}
If $f\in B$, we have
\begin{equation}\label{eq820}
\|E^cf\|_{L^2}^2=(2\pi)^{-n}\|F_\pm f\|_{L^2}^2.
\end{equation}
Thus the maps $f\mapsto F_\pm f$ can be extended for $f\in L^2$ satisfying \eqref{eq820}, which vanish on $E^dL^2$ and restrict to isometries $E^cL^2\rightarrow L^2(d\xi/(2\pi)^n)$. Moreover, we have the intertwining property
\begin{equation}\label{eq821}
F_\pm e^{itH}f=e^{it|\cdot|^s}F_\pm f,\quad f\in L^2,~t\in\mathbb{R}.
\end{equation}
\end{theorem}

\begin{proof}
\eqref{eq820} follows from \eqref{eq88} and Lemma \ref{lm82}. For \eqref{eq821}, we first consider $f\in\mathscr{F}^{-1}C_c^\infty(\mathbb{R}^n\setminus\{0\})$. For any $\lambda\in\mathbb{R}\setminus\tilde{\Lambda}$, we have $(H_0-\lambda)f\in\mathscr{S}$ and thus $R_0(\lambda\pm i0)(H_0-\lambda)f=f$ by weak* continuity of $R_0$ and dominated convergence. It follows that
\begin{equation*}
(H-\lambda)f=(I+VR_0(\lambda\pm i0))(H_0-\lambda)f,
\end{equation*}
which implies
\begin{equation*}
(F_\pm(H-\lambda)f)(\xi)=(|\xi|^s-\lambda)\hat{f}(\xi)=0,\quad\xi\in M_\lambda,
\end{equation*}
and this just means
\begin{equation}\label{eq824}
F_\pm Hf=|\cdot|^sF_\pm f.
\end{equation}
Since $\mathscr{F}^{-1}C_c^\infty(\mathbb{R}^n\setminus\{0\})$ is dense in $H^s$ and $H$ is closed, \eqref{eq824} holds when $f\in H^s$, which in the meanwhile shows that $F_\pm f\in\mathrm{Dom}(|\cdot|^s)$. Now we can differentiate $e^{-it|\cdot|^s}F_\pm e^{itH}f$ in $t$ when $f\in H^s$, and show that the derivative is $0$, therefore \eqref{eq821} holds for $H^s$ is dense in $L^2$.
\end{proof}

\begin{theorem}\label{thm74}
$F_\pm W_\pm=\mathscr{F}$ holds in $L^2$.
\end{theorem}

\begin{proof}
We just show for the minus sign. If $f\in\mathscr{F}^{-1}C_c^\infty(\mathbb{R}^n\setminus\{0\})$, since $e^{-itH_0}f\subset\mathscr{F}^{-1}C_c^\infty(\mathbb{R}^n\setminus\{0\})$ is infinitely differentiable in $L^2$, and we can differentiate $e^{itH}e^{-itH_0}f$ to get
\begin{equation}\label{eq825}
W_-f=f-\int_{-\infty}^0e^{itH}iVe^{-itH_0}fdt.
\end{equation}
Since $F_-$ is $L^2$ continuous, we can use the intertwining property \eqref{eq821} and the fact that the integral in \eqref{eq825} is absolutely convergent in $L^2$ by \eqref{eq53}, to deduce
\begin{equation}\label{eq826}
\begin{split}
F_-W_-f=&F_-f-\int_{-\infty}^0e^{it|\cdot|^s}F_-(iVe^{-itH_0}f)dt\\
=&F_-f-\lim_{\epsilon\downarrow0}\int_{-\infty}^0e^{\epsilon t+it|\cdot|^s}F_-(iVe^{-itH_0}f)dt.
\end{split}
\end{equation}
Now we take any $\phi\in C_c^\infty(\mathbb{R}^n\setminus\cup_{\lambda\in\tilde{\Lambda}}M_\lambda)$, and assume that $\mathrm{supp}\,\phi\subset I\times\mathbb{S}^{n-1}$ where compact $I\subset\mathbb{R}_+\setminus\tilde{\Lambda}$. Note that $\lambda\geq c>0$ if $\lambda\in I$. We have
\begin{equation*}
\begin{split}
\left\langle\int_{-\infty}^0e^{\epsilon t+it|\cdot|^s}F_-(iVe^{-itH_0}f)dt,\phi\right\rangle=&\int_{-\infty}^0\left\langle e^{\epsilon t+it|\cdot|^s}F_-(iVe^{-itH_0}f),\phi\right\rangle dt\\
=&\int_{-\infty}^0dt\int_{\lambda^\frac1s\in I}\frac{d\lambda}{s\lambda^\frac{s-1}{s}}\int_{M_\lambda}F_-(e^{\epsilon t+it\lambda}iVe^{-itH_0}f)\bar{\phi}dS.
\end{split}
\end{equation*}
Since $\lambda\in\mathbb{R}\setminus\tilde{\Lambda}$, by Lemma \ref{lm82} and Lemma \ref{lm23} we have
\begin{equation}\label{eq828}
\begin{split}
\|F_-(e^{\epsilon t+it\lambda}iVe^{-itH_0}f)\|_{L^2(M_\lambda)}\leq&C(\lambda)e^{\epsilon t}\|(I+VR_0(\lambda-i0))^{-1}Ve^{-itH_0}f\|_B\\
\leq&C'(\lambda)e^{\epsilon t}\|V\|_{B_s^*\rightarrow B}\|J_sf\|_{L^2}\\
\leq&C''(\lambda)e^{\epsilon t},
\end{split}
\end{equation}
where $C''(\lambda)$ is locally bounded. Indeed, the fact that $C(\lambda)$ is locally bounded rather comes from the proof of Lemma \ref{lm23} in our specific case (see \cite[remark of Theorem 2.1]{AH}); and we have used the local boundedness of $\|(I+VR_0(\lambda-i0))^{-1}\|_{B\rightarrow B}$ in $\lambda$, which is due to Lemma \ref{lm81} and the Banach-Steinhaus theorem. Now we can use the Fubini's theorem to exchange integrations and obtain
\begin{equation*}
\begin{split}
&\left\langle\int_{-\infty}^0e^{\epsilon t+it|\cdot|^s}F_-(iVe^{-itH_0}f)dt,\phi\right\rangle\\
=&\int_{\lambda^\frac1s\in I}\frac{d\lambda}{s\lambda^\frac{s-1}{s}}\int_{-\infty}^0dt\int_{M_\lambda}F_-(e^{\epsilon t+it\lambda}iVe^{-itH_0}f)\bar{\phi}dS\\
=&\int_{\lambda^\frac1s\in I}\frac{d\lambda}{s\lambda^\frac{s-1}{s}}\int_{M_\lambda}\left(\int_{-\infty}^0F_-(e^{\epsilon t+it\lambda}iVe^{-itH_0}f)\big|_{M_\lambda}dt\right)\bar{\phi}dS,
\end{split}
\end{equation*}
where the last line also comes from \eqref{eq828}, and $F_-(e^{\epsilon t+it\lambda}iVe^{-itH_0}f)\big|_{M_\lambda}$ denotes the $L^2$ trace on $M_\lambda$. For fixed $\lambda$, the second estimate of \eqref{eq828} implies that $\int_{-\infty}^0e^{\epsilon t+it\lambda}iVe^{-itH}fdt$ is absolutely convergent in $B$; further, the map $u\mapsto F_-u|_{M_\lambda}$ is continuous from $B$ to $L^2(M_\lambda)$. We thus conclude that
\begin{equation*}
\int_{-\infty}^0F_-(e^{\epsilon t+it\lambda}iVe^{-itH_0}f)\arrowvert_{M_\lambda}dt=F_-\left(\mbox{$\int_{-\infty}^0e^{\epsilon t+it\lambda}iVe^{-itH_0}fdt$}\right)\Bigg|_{M_\lambda}.
\end{equation*}
Therefore,
\begin{equation*}
\left\langle\int_{-\infty}^0e^{\epsilon t+it|\cdot|^s}F_-(iVe^{-itH_0}f)dt,\phi\right\rangle=\int_{\lambda^\frac1s\in I}\frac{d\lambda}{s\lambda^\frac{s-1}{s}}\int_{M_\lambda}F_-\left(\mbox{$\int_{-\infty}^0e^{\epsilon t+it\lambda}iVe^{-itH_0}fdt$}\right)\bar{\phi}dS.
\end{equation*}
Also notice that $V$ is continuous from $H^s$ to $L^2$, and that $\int_{-\infty}^0e^{\epsilon t+it\lambda}J_s(e^{-itH_0}f)dt$ is absolutely convergent in $L^2$, we have
\begin{equation*}
\begin{split}
\left\langle\int_{-\infty}^0e^{\epsilon t+it|\cdot|^s}F_-(iVe^{-itH_0}f)dt,\phi\right\rangle=&\int_{\lambda^\frac1s\in I}\frac{d\lambda}{s\lambda^\frac{s-1}{s}}\int_{M_\lambda}F_-V\left(\mbox{$\int_{-\infty}^0e^{\epsilon t+it\lambda}ie^{-itH_0}fdt$}\right)\bar{\phi}dS\\
=&\int_{\lambda^\frac1s\in I}\frac{d\lambda}{s\lambda^\frac{s-1}{s}}\int_{M_\lambda}(F_-VR_0(\lambda-i\epsilon)f)\bar{\phi}dS.
\end{split}
\end{equation*}
Combining this and \eqref{eq826},
\begin{equation*}
\begin{split}
&\left|\langle F_-W_-f-\hat{f},\phi\rangle\right|\\
\leq&\lim_{\epsilon\downarrow0}\int_{\lambda^\frac1s\in I}\frac{1}{s\lambda^\frac{s-1}{s}}\left|\int_{M_\lambda}\mathscr{F}\left((I+VR_0(\lambda-i0))^{-1}(I+VR_0(\lambda-i\epsilon))f-f\right)\bar{\phi}dS\right|d\lambda\\
\leq&\lim_{\epsilon\downarrow0}C_{I,\phi}\int_{\lambda^\frac1s\in I}\left\|(I+VR_0(\lambda-i0))^{-1}(I+VR_0(\lambda-i\epsilon))f-f\right\|_Bd\lambda\\
=&0,
\end{split}
\end{equation*}
for $(I+VR_0(\lambda-i0))^{-1}(I+VR_0(\lambda-i\epsilon))f$ is continuous both in $\epsilon$ and in $\lambda$ by Lemma \ref{lm76} and Lemma \ref{lm81}. Finally, since $\mathscr{F}^{-1}C_c^\infty(\mathbb{R}^n\setminus\{0\})$ and $C_c^\infty(\mathbb{R}^n\setminus\cup_{\lambda\in\tilde{\Lambda}}M_\lambda)$ are both dense in $L^2$, the proof is complete.
\end{proof}

We are ready to prove Theorem \ref{thm14} now.

\begin{proof}[Proof of Theorem \ref{thm14}]
Since $H_0$ is absolutely continuous, a general fact (see e.g. \cite[p. 531]{Kato}) is that $\mathrm{Ran}W_\pm\subset\mathscr{H}_\mathrm{ac}$. Now Theorem \ref{thm73} and Theorem \ref{thm74} indicate that $E^cL^2=\mathrm{Ran}W_\pm$, therefore $E^cL^2\subset\mathscr{H}_\mathrm{ac}$ also holds. On the other hand, 
by (iii) of Proposition \ref{pro72}, we know that $\tilde{\Lambda}$ is countable and thus
\begin{equation}\label{equ622}
E^dL^2\subset\mathscr{H}_\mathrm{pp},
\end{equation}
the pure point subspace of $L^2$ with respect to $H$, which implies $E^cL^2\supset\mathscr{H}_\mathrm{ac}\oplus\mathscr{H}_\mathrm{sing}$. We then must have \eqref{equ18}.
\end{proof}

By the proof of Theorem \ref{thm14}, it follows that $E^dL^2=\mathscr{H}_\mathrm{pp}$, which combining with Remark \ref{rk65} implies
\begin{equation}\label{e81}
\Lambda=\sigma_\mathrm{pp}\setminus\{0\}.
\end{equation}


\section{Decay and Regularity of Eigenfunctions}\label{sec9}

In this section, we will prove through Fredholm theory that actually every eigenfunction $u$ of $H$ associated with $\lambda\in\Lambda=\sigma_\mathrm{pp}\setminus\{0\}$ has the form suggested by Remark \ref{rk62}, which thus proves the other part of Theorem \ref{thm13}. As indicated in Section \ref{sec7}, we shall study the equation \eqref{equ61}, and the main result is Proposition \ref{prop74}. For $\lambda\in\mathbb{R}\setminus\{0\}$, recall that $M_\lambda=\{\xi\in\mathbb{R}^n;~|\xi|^s=\lambda\}$.

\begin{lemma}\label{lm73}
Let $\lambda\in\mathbb{R}\setminus\{0\}$, $V$ be a short range potential, and $u\in B_s^*$ satisfies $((-\Delta)^\frac s2-\lambda)u=-Vu$ (in the sense of \eqref{equ62}). Then
	\begin{equation}
	\langle u,f+VR_0(\lambda+i0)f\rangle=0,\quad f\in B,\,\,\hat{f}=0\,\,\text{on}\,\,M_\lambda.
	\end{equation}
\end{lemma}

\begin{proof}
By Lemma \ref{lm72}, we can write
\begin{equation}\label{eq76}
u=-R_0(\lambda-i0)Vu+u_-,\quad u_-\in\mathscr{F}^{-1}L^2(M_\lambda,dS).
\end{equation}
In particular, this implies that $u_-\in B^*$. Then
\begin{equation}\label{eq77}
\langle u,f\rangle=-\langle R_0(\lambda-i0)Vu,f\rangle+\langle u_-,f\rangle.
\end{equation}
For the first term on the right hand side above, by the weak* continuity of $\mathbb{C}^+\ni z\mapsto R_0(z)Vu\in B^*$ and the fact that $f\in B$, we have
\begin{equation*}
\begin{split}
-\langle R_0(\lambda-i0)Vu,f\rangle=&-\lim_{\epsilon\downarrow0}\int(R_0(\lambda-i\epsilon)Vu)\bar{f}dx\\
=&-\langle u,VR_0(\lambda+i0)f\rangle.
\end{split}
\end{equation*}
	The second term on the right hand side of \eqref{eq77} is $0$, since $\hat{f}$ has trace $0$ on $M_\lambda$, and we can use Lemma \ref{lm23} and \eqref{eq76} by taking $\mathscr{S}\ni f_j\rightarrow f$ in $B$.
\end{proof}

The main result of the section is the following.

\begin{proposition}\label{prop74}
Suppose $\lambda\in\mathbb{R}\setminus\{0\}$, $V$ is a short range potential, and $A_\lambda\neq\{0\}$ where
\begin{equation}
A_\lambda=\left\{f\in B;~\mathscr{F}\left((VR_0(\lambda+i0))^kf\right)=0\text{\,\,on\,\,}M_\lambda,\,k=0,1,2,\cdots\right\}.
\end{equation}
If $u\in B_s^*$ satisfies $((-\Delta)^\frac s2-\lambda)u=-Vu$ (in the sense of \eqref{equ62}), then
\begin{equation}\label{eq79}
u=-R_0(\lambda+i0)Vu=-R_0(\lambda-i0)Vu.
\end{equation}
\end{proposition}

\begin{proof}
We first note that $A_\lambda$ is a well defined closed subspace of $B$ which can be seen by Theorem \ref{thm34} and Lemma \ref{lm23}. The assumption $A_\lambda\neq\{0\}$ guarantees the non-trivial inclusion $B^*\subset A_\lambda^*$, where $A_\lambda^*$ is the $B$-norm dual space of $A_\lambda$. Consider operator $T=I+VR_0(\lambda+i0)$ on $A_\lambda$. Obviously $VR_0(\lambda+i0)$ is well defined and compact in $A_\lambda$ with the $B$ topology, thus by Fredholm theory,
	\begin{equation}\label{eq710}
	\mathrm{dim}\,\mathrm{Ker}(T)=\mathrm{dim}\,({}^\perp\mathrm{Ran}(T))<\infty,
	\end{equation}
	where
	\begin{equation}\label{eq711}
	{}^\perp\mathrm{Ran}(T)=\left\{v\in A_\lambda^*;~\langle v,Tf\rangle,\,\,f\in A_\lambda\right\},
	\end{equation}
	and here $\langle\cdot,\cdot\rangle$ denotes the duality pairing between $A_\lambda^*$ and $A_\lambda$. Let $f_1,\cdots,f_r\in A_\lambda$ be a basis of $\mathrm{Ker}(T)$, and set
	\begin{equation}
	u_j=R_0(\lambda+i0)f_j\in B_s^*\subset B^*\subset A_\lambda^*.
	\end{equation}
	Now one checks that in the sense of \eqref{equ62},
	\begin{equation}
	((-\Delta)^\frac s2-\lambda)u_j=f_j=-Vu_j,
	\end{equation}
	which implies that $u_1,\cdots,u_r$ are linearly independent. When $f\in A_\lambda$, $\langle u_j,Tf\rangle$ is then equal to the $B^*-B$ pairing \eqref{e115} in the sense of isometry, which is $0$ by Lemma $\ref{lm73}$, and we know from \eqref{eq710} and \eqref{eq711}  that ${}^\perp\mathrm{Ran}(T)=\mathrm{span}\{u_1,\cdots,u_r\}$. We also have $u\in{}^\perp\mathrm{Ran}(T)$ for the same reason, thus $u$ is a linear combination of $u_1,\cdots,u_r$ and must satisfy the first equality of \eqref{eq79}, where the second equality follows by the last statement of Lemma \ref{lm72}.
\end{proof}

Finally, we are ready to illustrate the proof of Theorem \ref{thm13}.

\begin{proof}[Proof of Theorem \ref{thm13}]
We first claim that if $\lambda\in\mathbb{R}\setminus\{0\}$ and $A_\lambda=\{0\}$, then $I+VR_0(\lambda+i0)$ is invertible in $B$. Otherwise, take $0\neq f\in B$ such that $f+VR_0(\lambda+i0)f=0$, it is easy to see by Remark \ref{rk62}, Lemma \ref{lm72} and \eqref{eq73} that $f\in A_\lambda$, which is a contradiction.
	
Therefore, if $\lambda\in\Lambda$ then $A_\lambda\neq\{0\}$ holds by definition, and if $u\in H^s$ satisfies $Hu=\lambda u$, we know from \eqref{eq79} that Proposition \ref{pro72} is applicable, \eqref{eq16} thus follows. The discreteness of $\sigma_\mathrm{pp}\setminus\{0\}$ is implied by \eqref{e81} and Proposition \ref{pro72}. The proof is complete.
\end{proof}

\noindent
\section*{Acknowledgements}
Tianxiao Huang is supported by the Fundamental Research Funds for the Central Universities (Sun Yat-sen University No. 19lgpy249), the China Postdoctoral Science Foundation (No. 2020M672929), and the Guangdong Basic and Applied Basic Research Foundation (No. 2020A1515111048). Tianxiao Huang thanks Ze Li, Kuijie Li, Mingjuan Chen and Minjie Shan for some useful discussion. Quan Zheng is supported by the National Natural Science Foundation of China (No. 11801188). The authors also thank Shanlin Huang for pointing out some research on the non-existence of wave operators.

\end{document}